\documentclass[journal]{IEEEtran}

\usepackage[cmex10]{amsmath}
\allowdisplaybreaks[4]

\usepackage{enumitem}
\usepackage{graphicx}
\usepackage{color}
\usepackage{xcolor}
\usepackage{amsmath}
\usepackage{mathtools}
\usepackage{multicol}
\usepackage{multirow}
\usepackage[english]{babel}
\usepackage{blindtext}
\usepackage{algorithm}
\usepackage{algorithmic}
\usepackage{balance}
\usepackage{amsfonts}
\usepackage{bm}
\usepackage{stfloats}
\usepackage{subfig}
\usepackage{amsthm}
\usepackage{amssymb}
\usepackage{setspace}
\usepackage[nosort]{cite}
\usepackage{CJK}
\usepackage{cite}
\usepackage{caption}
\usepackage{comment}
\usepackage{array,multirow}
\usepackage{graphicx}
\newtheorem{lemma}{Lemma}
\newtheorem{theorem}{Theorem}

\newtheorem{remark}{Remark}
\newtheorem{proposition}{Proposition}
\theoremstyle{plain}
\newtheorem{definition}{Definition}

\usepackage[table]{xcolor}

\begin{document}

\captionsetup[figure]{labelformat={default},labelsep=period,name={Fig.}}
\captionsetup[table]{labelformat={default},labelsep=period,name={Table}}

\title{Fundamentals of NOMA in Low-Earth Orbit Coordinated Multi-Satellite Networks}

\author{Xiangyu Li,
        Bodong Shang,~\IEEEmembership{Member,~IEEE},
        Junchao Ma,
        Qingqing Wu,~\IEEEmembership{Senior Member,~IEEE},
        Jie Feng,~\IEEEmembership{Member,~IEEE}, 
        and Deshuang Huang,~\IEEEmembership{Fellow,~IEEE}
\thanks{This work was supported in part by the Zhejiang Provincial Natural Science Foundation of China under Grant No. LQN25F010003, in part by the Ningbo Natural Science Foundation under grant 2025J021, in part by the State Key Laboratory of Integrated Services Networks, Xidian University, in part by the Fundamental Research Funds for the Central Universities under Grant QTZX26142, in part by NSFC under grant 62401233, in part by the Natural Science Foundation of Jiangsu Province under grant BK20241076, and in part by the YongRiver Scientific and Technological Innovation Project No. 2023A-187-G. \textit{(Corresponding author: Bodong Shang)}}
\thanks{Xiangyu Li and Bodong Shang are with the Zhejiang Key Laboratory of Industrial Intelligence and Digital Twin, Eastern Institute of Technology, Ningbo, Zhejiang 315200, China, also with the State Key Laboratory of Integrated Services Networks, Xidian University, Xi’an 710071, China, and also with the Department of Electronic Engineering, Shanghai Jiao Tong University, Shanghai 200240, China (e-mail: xyli@eitech.edu.cn; bdshang@eitech.edu.cn).}
\thanks{Junchao Ma is with the School of Electrical and Information Engineering, Jiangsu University of Technology, Changzhou 213001, China (e-mail: junchao\_ma@jstu.edu.cn).}
\thanks{Qingqing Wu is with the Department of Electronic Engineering, Shanghai Jiao Tong University, Shanghai 200240, China (e-mail: qingqingwu@sjtu.edu.cn).}
\thanks{Jie Feng is with the School of Telecommunications Engineering, Xidian University, Xi’an 710071, China (e-mail: fengjie@xidian.edu.cn).}
\thanks{Deshuang Huang is with the College of Information Science and Technology, Eastern Institute of Technology, Ningbo, Zhejiang 315200, China (e-mail: dshuang@eitech.edu.cn).}
}

\markboth{IEEE Transactions on Green Communications and Networking}%
{Li \MakeLowercase{\textit{et al.}}: Fundamentals of NOMA in Low-Earth Orbit Coordinated Multi-Satellite Networks}

\maketitle

\begin{abstract}
Coordinated multi-satellite (CoMS) transmission and non-orthogonal multiple access (NOMA) are envisioned to jointly enhance coverage, capacity, and spectrum efficiency for satellite networks. Their integration into a unified CoMS-NOMA framework will allow more efficient, reliable, and energy-efficient multi-user access.  
This paper investigates the downlink performance of CoMS-NOMA networks from a system-level perspective, in which multiple satellites cooperatively serve multiple users via NOMA. 
Leveraging tools from stochastic geometry, related angles and distances in CoMS-NOMA are first derived as intermediate results. Then, we obtain the combined signal power distributions and analyze coverage and spectrum performance under both inter- and intra-satellite interference, accounting for potential imperfect successive interference cancellation (SIC). The analytical model is validated across a range of system parameters, including the number of satellites, service region angle, error-propagation factor, and power allocation coefficients.
Numerical results indicate that increasing the number of cooperative satellites does not always improve coverage and spectrum efficiency. 
Additionally, while a higher main-lobe gain improves coverage, a near-perfect SIC provides only slightly greater benefits than a reasonably good SIC.
With properly selected power allocation coefficients, CoMS-NOMA achieves up to a 270\% improvement in coverage and a 56\% gain in sum spectral efficiency, compared with conventional orthogonal and single-satellite schemes, indicating potential for green, energy-efficient satellite networking.
\end{abstract}

\begin{IEEEkeywords}
Coordinated transmission, low-Earth orbit (LEO) satellite, non-orthogonal multiple access (NOMA), coverage probability, and spectral efficiency.
\end{IEEEkeywords}

\IEEEpeerreviewmaketitle

\section{Introduction}

\IEEEPARstart{I}{n} recent years, satellite communication (SatCom) has become an indispensable component of next-generation wireless systems, providing global coverage and ubiquitous connectivity for broadband and Internet of Things (IoT) applications \cite{su2019broadband,ding2024improving}. 
Depending on their functions and deployment locations, satellites operate in the low-Earth orbit (LEO), medium-Earth orbit (MEO), and geostationary Earth orbit (GEO). Among them, LEO satellites are attracting significant interest from researchers in academia and industry \cite{li2025advancing}. 
Due to their high data rates, reduced latency, comparatively lower power consumption, and maintenance costs, LEO satellites are rapidly developing for large-scale deployments to offer direct-to-cell services \cite{peng2026enhanced} in the terrestrial network (TN), especially in remote and rural areas.

However, while the number of LEO satellites is increasing, the user density and service diversity continue to grow explosively \cite{luo2024leo}. 
This growth intensifies competition for limited spectrum resources and heightens the demand for dynamic inter-satellite coordination. 
Consequently, severe challenges arise in the efficient utilization of spectrum resources and the coordinated management of multi-satellite transmission.
From end users’ viewpoint, inefficient physical-layer transmission and interference management may lead to degraded communication reliability, increased packet loss, repeated retransmissions, and larger service delays \cite{wang2026packet}, particularly for latency-sensitive or connectivity-critical applications. In dense LEO satellite networks, severe inter-satellite interference and inefficient spectrum utilization can further reduce user quality of service and limit the capability of the network for massive simultaneous access.

\vspace{-4mm}
\subsection{Related Works}
Traditional orthogonal multiple access (OMA) schemes, which allocate distinct time, frequency, or spatial resources to different users, have been widely applied by default in most SatCom works, such as \cite{okati2020downlink,park2022tractable,jia2022analytic}, which provided downlink or uplink analysis in global large-scale satellite to TN user terminals (UTs) communications. 
While intra-satellite interference-free transmission can be ensured, satellite OMA suffers from limited spectral efficiency and may struggle to accommodate massive connectivity and high-throughput demands.
This limitation is especially significant for UTs in remote or underserved regions, where multiple UTs may simultaneously request urgent connections. In such situations, a single satellite (SglS) with OMA schemes may not be sufficient to serve these UTs effectively.

To improve spectrum utilization for SatCom, non-orthogonal multiple access (NOMA) has emerged as a promising technique, enabling multiple UTs to share the same resource block (RB) through power-domain multiplexing and successive interference cancellation (SIC).
To explore the effects of NOMA on a typical serving satellite, the authors in \cite{li2024performance} and the authors in \cite{yang2025performance} compared NOMA-aided SatCom by setting a near-user and a far-user to demonstrate the superiority of NOMA over OMA. 
By extending NOMA to multi-user transmission scenarios, \cite{li2026coverage} provided a new level of theoretical analysis to investigate the impact of the number of UTs on system spectral efficiency; however, this analysis relied only on a typical satellite-based service. 
Despite the potentially increased system decoding complexity, multi-user scenarios can be of practical significance as the number of UTs to be served grows.
In addition, the authors in \cite{gao2020performance} analyzed the downlink performance of an LEO satellite employing NOMA to serve two UTs under Doppler shift effects. The impact of hardware impairments on downlink NOMA and cognitive radio transmission in satellite systems was examined in \cite{Akhmetkaziyev2021performance}. 
Furthermore, the authors in \cite{khan2024ris} studied reconfigurable intelligent surface (RIS)-assisted LEO satellite communications with NOMA from an energy-efficient aspect, and the authors in \cite{asif2025transmissive} took a step further by applying transmissive RIS using hybrid NOMA. 
It is worth mentioning that these works and recent studies on satellite NOMA are limited to scenarios involving an SglS setting for service provision, thereby neglecting the potential benefits of multi-satellite cooperation.

To address the limitations of the SglS service architecture, the coordinated multi-satellite (CoMS) network, or the multi-satellite cooperative communication (MSCC) system \cite{shang2026multi,pan2026multi}, has recently attracted researchers' attention. 
In \cite{li2024analytical}, an analytical CoMS joint transmission model was proposed to analyze coverage and data rate in a harsh and complex fading environment. The authors of \cite{qu2025coverage} further studied an MSCC system employing a two-state Markov channel model and determined the optimal satellite density to maximize coverage probability. 
Due to the inherent characteristics of SatCom networks, the number of cooperative satellites and their orbital altitudes can be significant for real-world deployment, and this choice involves a trade-off to achieve the optimal coverage probability, as shown in \cite{shang2023coverage}.
To facilitate cooperative satellite networking, the concept of cell-free massive multiple-input multiple-output (CF-mMIMO) has recently been introduced to SatCom. The authors in \cite{abdelsadek2022distributed} maximized network throughput and minimized handover rate by employing satellite-UT transmission in a CF-mMIMO manner. The benefits of such satellite cooperation were later studied in \cite{abdelsadek2023broadband}, demonstrating its superiority over single-satellite (SglS) connectivity.
The above works relied mainly on coherent joint transmission, which requires accurate channel state information (CSI) and tight phase synchronization among multiple satellites \cite{tanbourgi2014tractable}. 
However, the high mobility of LEO satellites and the rapid, time-varying nature of their channels make it challenging to maintain synchronization and reliably acquire CSI \cite{baeza2022non,wang2025non}.
To address this issue, a non-coherent joint transmission (NC-JT) strategy can be adopted for CoMS transmission \cite{d2024coherent,wang2023non}. In this case, cooperating satellites jointly transmit the same data symbols to the target UTs without prior phase alignment, and the received non-coherent signal components combine constructively to enhance the overall received power.

Building upon the respective advantages of NOMA and CoMS transmission, their integration, termed \textit{CoMS-NOMA}, offers a new dimension for enhancing the overall performance of SatCom networks.
Specifically, for the target UTs, NOMA enhances spectral efficiency and user multiplexing within each satellite, while CoMS provides cooperative transmission across satellites.
Their interplay enables a two-fold gain: intra-satellite NOMA multiplexing gain and inter-satellite coordination gain.
Moreover, multi-user NOMA under the CoMS architecture introduces new interactions between power-domain multiplexing and inter-satellite coordination, distinct from conventional two-user NOMA settings. 
The work in \cite{li2019non} investigated a similar CoMS-NOMA architecture in terms of ergodic sum rate; however, it relied on free-space path-loss and Rayleigh small-scale fading assumptions, which were not representative of realistic SatCom environments.
In \cite{elhalawany2022outage}, while practical fading models were incorporated, several fundamental characteristics of the CoMS system, such as the number of satellites, orbital altitude, and antenna lobe gain, were neglected. 
Additionally, both works considered only a limited cooperative configuration with two serving satellites and two UTs served by each satellite, without accounting for interfering satellites, which significantly restricts the practicality and scalability of the CoMS-NOMA network. Therefore, large-scale analysis of CoMS-NOMA with multiple UTs remains largely unexplored.

Stochastic geometry is a powerful mathematical tool for characterizing and analyzing satellite networks with irregular topologies.
Many recent works have explored LEO satellite networks using this tool \cite{okati2020downlink,park2022tractable,jia2022analytic}.
For example, the authors in \cite{okati2020downlink} studied the downlink coverage and rate of LEO satellite constellations by modeling satellite locations as a binomial point process (BPP) in a finite space.
However, to improve tractability, the Poisson point process (PPP) has been regarded as a preferred choice. 
In \cite{park2022tractable}, a tractable downlink model was investigated to capture the effects of orbital altitude, satellite density, and fading parameters under Nakagami fading. It showed that there exists an optimal number of satellites at certain orbital altitudes.
Moreover, to provide guidelines for mega-constellation deployment, the authors in \cite{jia2022analytic} developed an analytical approach to assess uplink performance using a Nakagami-$m$ approximated Shadowed-Rician fading model.

\subsection{Contributions and Paper Organizations}
Motivated by the above observations, we note that the key analytical challenge lies in the joint characterization of CoMS transmission and NOMA decoding with multiple UTs in a system-level framework. 
The studied CoMS-NOMA framework is suitable for dense LEO satellite networks requiring enhanced coverage reliability and spectral efficiency, especially in remote regions with increasing user access, coverage-edge areas with limited signal quality, and delay-sensitive NTN services. By leveraging cooperative multi-satellite transmission and power-domain multiplexing, CoMS-NOMA can improve resource utilization and communication robustness in these practical deployment scenarios.

Different from existing cooperative NOMA studies that mainly focus on terrestrial systems, SglS scenarios, or small-scale transmission analysis, this paper aims to develop a stochastic geometry-based system-level analytical framework for CoMS-NOMA transmission in large-scale LEO satellite networks. 
The studied framework jointly considers NC-JT-based cooperative transmission, inter-satellite interference, SIC operation, and cooperative multi-satellite coverage analysis with multiple UTs, providing tractable performance characterization for practical dense LEO satellite systems.

Specifically, the PPP constellation model is adopted to approximate practical satellite deployments and to enable analytically tractable results with broad applicability. 
A unified analytical framework is developed to characterize the interactions among multi-satellite coordination, multi-user NOMA transmission, inter-satellite interference, and other realistic factors.
The CoMS-NOMA framework studied provides valuable insights into how the system parameters influence network coverage and spectral efficiency.
The main contributions of this paper are summarized as follows:
\begin{itemize}
    \item \textit{CoMS-NOMA System Design:} We develop a CoMS-NOMA joint transmission framework for SatCom networks, incorporating practical factors such as inter-satellite interference, intra-satellite NOMA interference, SIC error propagation, and variations in satellite number, altitude, and antenna gain. The design aims to achieve a two-fold performance gain, i.e., multi-satellite coordination gain and NOMA multiplexing gain. Specifically, multiple satellites within the service region cooperatively serve the UTs, with each satellite employing NOMA transmission. This approach enables additional cooperative signal reinforcement and more efficient utilization of spectrum resources.
    \item \textit{Modeling and Analysis:} To assess system performance, an analytical framework is developed using stochastic geometry. Satellites are modeled on a spherical cap following a PPP manner at a common orbital altitude, while UTs are uniformly distributed in a serving area. The angle and distance distributions are first presented, followed by an analysis of the combined desired signals using the NJ-CT strategy and the aggregated interfering signals. Next, we derive theoretical expressions for the individual and average coverage probabilities, as well as the sum spectral efficiency of UTs.
    \item \textit{System Design Insights:} Extensive simulations and numerical analyses are conducted to quantitatively evaluate the CoMS-NOMA system and the impacts of key parameters. The results show that higher main-lobe gain and improved SIC accuracy are critical to enhancing coverage performance. Moreover, optimal service region angles and power allocation ratios among UTs are identified to maximize coverage probability or sum spectral efficiency. Comparative evaluations with the OMA and SglS schemes further demonstrate the superiority of CoMS-NOMA in some system configurations.
\end{itemize}

The remainder of this paper is organized as follows. 
Section \uppercase\expandafter{\romannumeral2} describes the system model for CoMS-NOMA networks, including the network model, channel model, and signal transmission. 
In Section \uppercase\expandafter{\romannumeral3}, we present key distributions regarding angles and distances. 
Section \uppercase\expandafter{\romannumeral4} explores the mathematical transformation of interfering signals and then analyzes the statistical characteristics of the combined desired signal power, followed by the derivation of coverage probability and spectral efficiency.
Simulations and numerical results are validated, and a discussion of key observations is provided in Section \uppercase\expandafter{\romannumeral5}.
Finally, Section \uppercase\expandafter{\romannumeral6} concludes the paper.
The main notations and related descriptions are summarized in Table \ref{symbol}.

\begin{table*}[t]
\setlength{\abovecaptionskip}{0cm} 
\setlength{\belowcaptionskip}{-0.0cm}
\captionsetup{font={scriptsize}}
\caption{Main Notations and Descriptions}
\label{symbol}
\centering
\small
\begin{tabular}{m{1.0cm}<{\centering} m{6.8cm}<{\centering} m{1.0cm}<{\centering} m{6.8cm}<{\centering}}
\hline
\hline
\textbf{Notation} & \textbf{Description} & \textbf{Notation} & \textbf{Description}  \\
\hline
\hline
$R_\text{S}$ & Radius of satellite orbit & 
$R_\text{E}$ & Radius of Earth \\
$R_\text{T}$ & TN serving area radius & 
$H_\text{S}$ & Altitude of satellites \\
$\lambda_\text{S}$ & Density of satellites & 
$K$ & Number of NOMA UTs \\
${{\mathcal{S}}_{\text{CS}}}$ & Set of cooperative satellites & 
${{\mathcal{I}}_{\text{IS}}}$ & Set of interfering satellites \\
$N_\text{CS}$ & Number of cooperative satellites & 
$N_\text{IS}$ & Number of interfering satellites \\
$G_\text{ml}$ & Satellite main-lobe gain & 
$G_\text{sl}$ & Satellite side-lobe gain \\
$\alpha$ & Path-loss coefficient & 
${L}_{0}$ & Path-loss on the reference distance \\
$m$ & Shape parameter of Nakagami distribution & 
$\Omega$ & Scale parameter of Nakagami distribution \\
$\kappa$ & Shape parameter of Gamma distribution & 
$\beta$ & Inverse scale parameter of Gamma distribution \\
$\varpi$ & Error propagation factor & 
$\omega_u$ & Additive white Gaussian noise at UT$_u$ \\
$P$ & Satellite transmit power & 
${\xi}_{u(,i)}$ & Power allocation coefficient for UT$_u$ (served by $i$-th cooperative satellite) \\
$f_c$ & Carrier frequency & 
$c$ & Speed of light \\
$q_i$ & Transmitted data symbol from $i$-th satellite &
$\{\cdot\}_{u,i}$ & Certain item between $u$-th UT and $i$-th satellite \\
${r}_{u,i}$ & Distance between $u$-th UT and $i$-th satellite & 
${g}_{u,i}$ & Small-scale fading coefficient \\
${\hat{h}}_{u,i}$ & Channel coefficient & 
${h}_{u,i}$ & Channel coefficient without ${L}_{0}$ \\
${{I}^{\text{inter}}}$ & Aggregated inter-satellite interference & 
${{I}^{\text{intra}}}$ & Aggregated intra-satellite interference \\
${\sigma }^{2}$ & Noise power & 
$\eta$ & Service region angle \\
$\varphi_{\text{d}}$ & Dome angle & 
$\varphi_{\text{z}}$ & Zenith angle \\
S & Point for a randomly-selected satellite & 
P & Projection point of S \\
O$_\text{T}$ & Point $(0,0,R_{\text{E}})$ & 
$\theta$ & Angle $\angle$UT$_{1}$-O$_\text{T}$-S \\
$r_{\text{N}}$ & Distance between point O$_\text{T}$ and S & 
${r}_{\text{SU}}$ & Distance between point S and UT$_1$ \\
${r}_{\text{T}}$ & Distance between point O$_\text{T}$ and UT$_1$ & 
$\zeta$ & Angle b/W $l_{\text{O}_\text{T}-\text{UT}_1}$ and projected line of $l_{\text{O}_\text{T}-\text{S}}$  \\
\hline
\hline
\end{tabular}
\vspace{-4mm}
\end{table*}

\section{System Model}

\subsection{Network Model}

Globally, LEO satellites are distributed around the Earth at the same altitude $H_\text{S}$ on a spherical shell surface according to a homogeneous PPP with a density $\lambda_\text{S}$. 
While practical satellite constellations follow deterministic orbital structures, the PPP-based model provides a tractable abstraction for characterizing the spatial statistics and interference behavior of large-scale dense LEO networks\footnote{Modeling satellites as a PPP has been widely adopted in stochastic geometry analyses of satellite networks \cite{park2022tractable,jia2022analytic,lee2025analyzing,wang2025modeling}. In this paper, the analytical results from the PPP-based model are further validated through simulations based on the Starlink constellation. Moreover, to follow deterministic orbital structures, current research is exploring the Cox point process (CPP) to approximate the LEO satellite constellation \cite{choi2024novel}, which can be left as future work for the CoMS-NOMA analysis.}.
Denote the radius of Earth as $R_\text{E}$, and the revolution radius of the satellites is calculated as $R_\text{S} = R_\text{E} + H_\text{S}$. 
Assume a typical circular serving area covered by satellites, within which a total of $K$ active UTs are uniformly distributed. While each satellite has its own nominal coverage region, this assumption captures the common practice of beam overlap in cooperative LEO constellations, where multiple satellites jointly serve the same coverage area.
Specifically, only a subset of UTs within the coverage region are scheduled for CoMS-NOMA transmission on a given resource block. As the service region angle or UT density increases, more candidate UTs coexist within the cooperative serving region, which reduces the scheduling probability of each target UT.

Denote the set of cooperative satellites for the UTs by ${{\mathcal{S}}_{\text{CS}}}$, and the set of interfering satellites by ${{\mathcal{I}}_{\text{IS}}}$.
Consider a downlink coordinated LEO satellite system, as shown in Fig. \ref{fig:system_model}. 
A number of $N_\text{CS} = \left| {{\mathcal{S}}_{\text{CS}}} \right|$ cooperative LEO satellites within the service region are coordinated to serve $K$ UTs by NOMA transmission, while a number of $N_\text{IS} = \left| {{\mathcal{I}}_{\text{IS}}} \right|$ remaining satellites are interfering satellites.
Note that serving satellites will not only deliver desired signals but also introduce intra-satellite interference, while interfering satellites lead to inter-satellite interference.

\captionsetup{font={scriptsize}}
\begin{figure}[h]
\begin{center}
\setlength{\abovecaptionskip}{+0.1cm}
\setlength{\belowcaptionskip}{-0.0cm}
\centering
  \includegraphics[width=3.3in, height=2.2in]{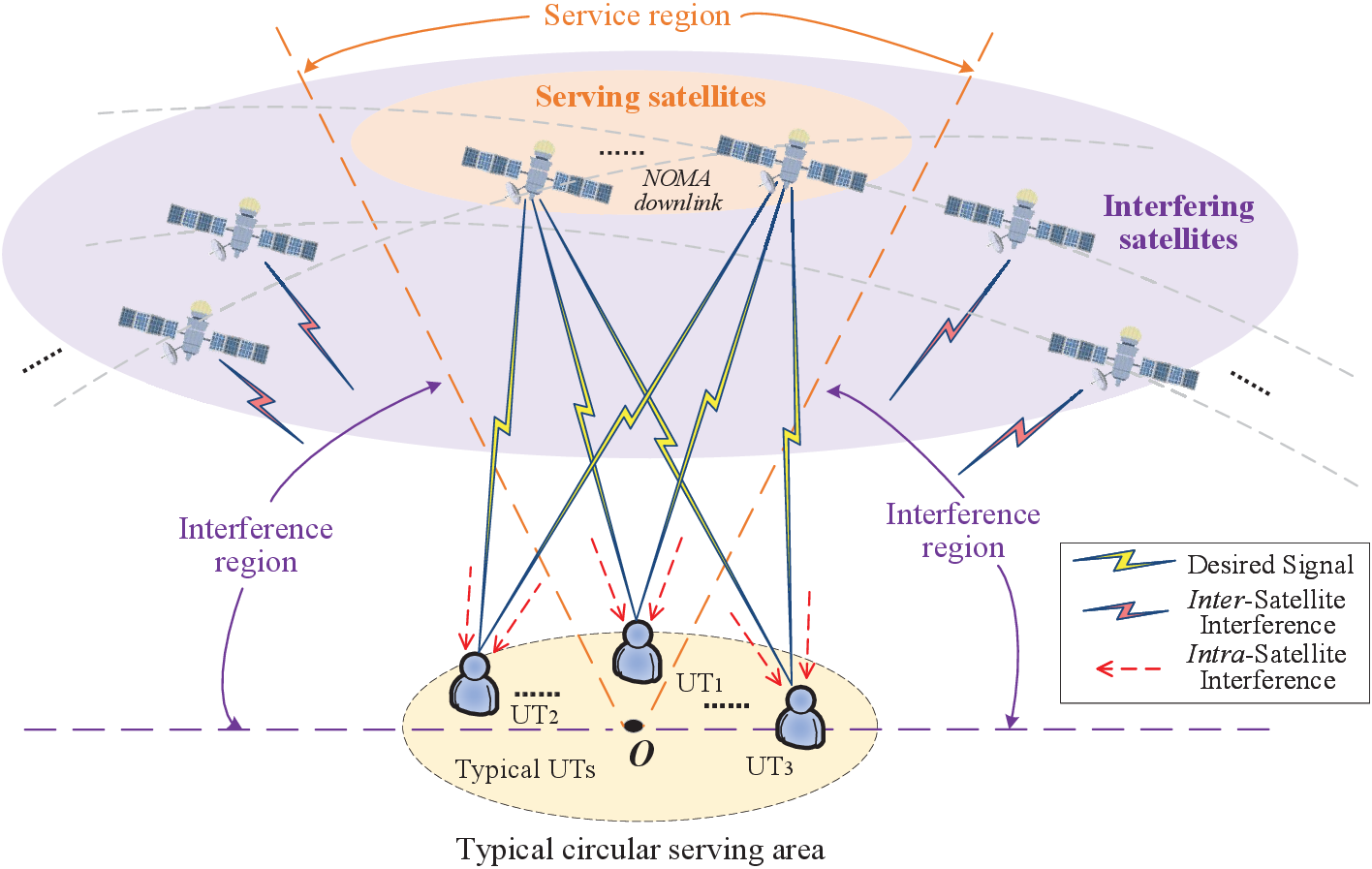}
\renewcommand\figurename{Fig.}
\caption{\scriptsize An illustration of the CoMS-NOMA network.}
\label{fig:system_model}
\end{center}
\vspace{-10mm}
\end{figure}

\subsection{Channel Model}
Channels between satellites and UTs experience both large-scale and small-scale fading. Specifically, the channel coefficient between the $i$-th satellite and the $u$-th UT, denoted as UT$_u$, can be written as
\begin{equation}
    {{\hat{h}}_{u,i}}
    = {{g}_{u,i}} L{{\left( {{r}_{u,i}} \right)}^{1/2}},
\end{equation}
where ${{g}_{u,i}}$ is the small-scale fading coefficient, and $L\left( {{r}_{u,i}} \right)={{L}_{0}}{{r}_{u,i}}^{-\alpha }$ represents the large-scale fading. $\alpha$ is the path-loss coefficient and ${{L}_{0}}={{\left( \frac{c}{4\pi {{f}_{c}}} \right)}^{2}}$ is the path-loss at a reference distance, with $c$ being the speed of light and $f_c$ being the carrier frequency.

To reflect the influence of line-of-sight (LoS) and non-line-of-sight (NLoS) components, the small-scale fading is assumed to follow an independent and identically distributed (i.i.d.) Nakagami-$m$ distribution\footnote{Note that while the Shadowed-Rician (SR) model has been established as a more accurate statistical characterization over Nakagami-$m$ fading for modeling satellite channels, its analytical tractability is limited by a complex PDF. This complexity often precludes the derivation of closed-form expressions, particularly when interference is taken into account.
Since it is shown in \cite{abdi2003new} that SR fading can be approximated by a Gamma variable through moment matching, the suboptimal yet mathematically tractable Nakagami-$m$ model remains a viable alternative to gain theoretical insights \cite{yang2025performance}.}.
Then, the probability density function (PDF) of $|{g}_{u,i}|$ is given by
\begin{equation}
    {{f}_{\left| {{g}_{u,i}} \right|}}\left( x \right)
    =\frac{2{{m}^{m}}}{\Gamma \left( m \right){{\Omega }^{m}}}{{x}^{2m-1}}{{e}^{-\frac{m}{\Omega }{{x}^{2}}}},
\end{equation}
where $x \geq 0$, $\Omega=\mathbb{E}\left\{\left|{g}_{u,i}\right|^2\right\}=1$, and the Gamma function $\Gamma(m)=(m-1)!$ for the integer $m>0$.
Consequently, $|{g}_{u,i}|^2$ can be regarded as a Gamma random variable, whose PDF is expressed as
\begin{equation}
    f_{\left|{{g}_{u,i}}\right|^2} \left( x \right) 
    = \frac{{{x}^{\kappa -1}}{{e}^{-\beta x}}{{\beta }^{\kappa }}}{\Gamma \left( \kappa  \right)}, 
\end{equation}
where $\kappa$ is a shape parameter and $\beta$ is an inverse scale parameter. Due to their square-law relationship \cite{nakagami1960the,holm2004sum}, $m=\kappa=\beta$. In the following, $m=\kappa$ and $\beta$ are written separately to capture their exact position.
It is worth noting that LEO satellite links typically experience large propagation delays and pronounced Doppler shifts due to the high relative velocity between satellites and UTs. In practical satellite communications, these effects are handled through physical-layer synchronization procedures. In particular, the propagation delay can be compensated via timing advance estimation mechanisms \cite{zhu2023timing} as also defined in the Third Generation Partnership Project (3GPP) standards, while Doppler shifts can be compensated through Doppler pre-compensation and carrier frequency offset (CFO) estimation and tracking \cite{yeh2024efficient} based on satellite ephemeris or GNSS assistance. After such synchronization and compensation, the residual delay and Doppler effects become negligible compared with large-scale path loss and small-scale fading, and the above channel modeling can be adopted for system-level performance analysis.

\subsection{CoMS-NOMA Signal}
In the following, we first present the expression of the aggregated signal received by multiple UTs from an SglS under NOMA transmission and then extend the formulation to the scenario with multiple cooperative satellites.

\subsubsection{SglS-NOMA Transmission}
For any one of the $i$-th cooperative satellite in ${{\mathcal{S}}_{\text{CS}}}$, the aggregated signal received at the $u$-th UT, denoted as UT$_u$, of the message intended for the $v$-th UT, denoted as UT$_v$, with $u\leq v\leq K$, is written as \cite{li2026coverage}
\begin{equation}
\begin{aligned}
    y_{u}^{v,i} 
    = & {\sqrt{{{G}_{\text{ml}}}}\sqrt{P{{\xi }_{u,i}}} {{\hat{h}}_{u,i}}{{q}_{i}}} \\
    & + {\sqrt{{{G}_{\text{ml}}}}\left( \sum\limits_{{{t}_{1}}=1}^{v-1}{\sqrt{P{{\xi }_{{{t}_{1}},i}}}}+\varpi \sum\limits_{{{t}_{2}}=v+1}^{K}{\sqrt{P{{\xi }_{{{t}_{2}},i}}}} \right) {{\hat{h}}_{u,i}}{{q}_{i}}} \\
    & + \sqrt{{{G}_{\text{sl}}}}\sum\limits_{j\in {{\mathcal{I}}_{\text{IS}}}}{\sqrt{P}{{\hat{h}}_{u,j}}{{q}_{j}}}+{{\omega }_{u}},
\end{aligned}
\end{equation}
where $P$ is the transmit power of the satellite, $G_\text{ml}$ and $G_\text{sl}$ are the main-lobe gain and the side-lobe gain, respectively, and $\varpi$ is the error propagation factor when applying NOMA. ${{\xi }_{u,i}}$ is the power allocation factor for UT$_u$ served by the $i$-th cooperative satellite, ${{\xi }_{u,i}} \in [0,1]$, $q_i$ is the transmitted data symbol from the $i$-th satellite, ${{\left| {{q}_{j}} \right|}^{2}}=1$, and $\omega_u$ is the additive white Gaussian noise (AWGN) at the UT receiver.

\begin{remark}
(Core Principle of SIC): The core principle of SIC is an iterative cycle of decoding, reconstructing, and subtracting (DRS). The process begins by treating the received composite signal as consisting of the strongest UT signal and interference from other sources. After this, the strongest signal is successfully decoded and re-encoded using estimates of the UT's channel and constellation. The reconstructed signal is then subtracted from the aggregate signal, which reduces interference for the subsequent decoding round targeting the next strongest user \cite{saito2013system}. 
However, since SIC is not perfect in practical NOMA operations, we adopt an error propagation factor $\varpi \in [0,1]$ to show the decoding conditions in a manner similar to \cite{sun2016non}. 
Particularly, $\varpi = 0$ represents the perfect decoding condition in which perfect SIC is assumed, while $\varpi = 1$ corresponds to the worst-case decoding condition with no SIC at all.
\end{remark}

Then, for an $i$-th cooperative satellite, the signal-to-interference-plus-noise ratio (SINR) at UT$_u$ of the message intended for UT$_v$ is expressed as \eqref{Formula_SlgS_NOMA_SINR_1} at the top of the page, where ${\sigma }^{2}$ is the noise power.
\begin{figure*}[!ht]
\setlength{\abovecaptionskip}{-1.5cm}
\setlength{\belowcaptionskip}{-0.5cm}
\normalsize
\begin{equation}
\begin{aligned}
    \text{SINR}_{u}^{v,i}
     & = \frac{{{G}_{\text{ml}}}P{{\xi }_{v,i}}{{\left| {{{\hat{h}}}_{u,i}} \right|}^{2}}}{{{G}_{\text{ml}}}P{{\left| {{{\hat{h}}}_{u,i}} \right|}^{2}}\left( \sum\limits_{{{t}_{1}}=1}^{v-1}{{{\xi }_{{{t}_{1}},i}}}+\varpi \sum\limits_{{{t}_{2}}=v+1}^{K}{{{\xi }_{{{t}_{2}},i}}} \right)+{{G}_{\text{sl}}} \sum\limits_{j\in {{\mathcal{I}}_{\text{IS}}}} {P{{\left| {{{\hat{h}}}_{u,j}} \right|}^{2}}}+{{\sigma }^{2}}},
\end{aligned}
\label{Formula_SlgS_NOMA_SINR_1}
\end{equation}
\hrulefill
\vspace{-4mm}
\end{figure*}
After mathematical reduction, \eqref{Formula_SlgS_NOMA_SINR_1} can be represented as 
\begin{equation}
\begin{aligned}
    \text{SINR}_{u}^{v,i}
     & = \frac{{{\xi }_{v,i}}{{\left| {{h}_{u,i}} \right|}^{2}}}{ I_{i}^{\text{intra}} + \sum\limits_{j\in {{\mathcal{I}}_{\text{IS}}}} I_{j}^{\text{inter}} + {{{\bar{\sigma }}}^{2}} },
\end{aligned}
\label{Formula_SlgS_NOMA_SINR_2}
\end{equation}
where ${{h}_{u,i}}={{r}_{u,i}}^{-\alpha /2}{{g}_{u,i}}$, and ${{\bar{\sigma }}^{2}}=\frac{{{\sigma }^{2}}}{{{G}_{\text{ml}}}P{{L}_{0}}}$. 
For ease of presentation, the term regarding intra-satellite interference from the $i$-th satellite, i.e., $I_{i}^{\text{intra}}$, $i \in {{\mathcal{S}}_\text{CS}}$, is denoted as $I_{i}^{\text{intra}} = {{\left| {{h}_{u,i}} \right|}^{2}}\left( \sum\limits_{{{m}_{1}}=1}^{v-1}{{{\xi }_{{{m}_{1}},i}}}+\varpi \sum\limits_{{{m}_{2}}=v+1}^{K}{{{\xi }_{{{m}_{2}},i}}} \right)$; 
the term regarding inter-satellite interference from the $j$-th satellite, i.e., $I_{j}^{\text{inter}}$, $j \in {{\mathcal{I}}_\text{IS}}$, is denoted as $I_{j}^{\text{inter}} = \frac{{{G}_{\text{sl}}}}{{{G}_{\text{ml}}}}{{\left| {{h}_{u,j}} \right|}^{2}}$.
It is shown that $\text{SINR}_{u}^{v,i}$ is almost independent of $P$ due to the relatively small noise power, while being highly dependent on ${{\xi }_{u,i}}$.

\subsubsection{CoMS-NOMA Transmission}
Considering multiple cooperative satellites, the aggregated signal received at UT$_u$ of the message intended for UT$_v$ is written as
\begin{equation}
\begin{aligned}
    y_{u}^{v} 
    = & \sum\limits_{i\in {{\mathcal{S}}_{\text{CS}}}}{\sqrt{{{G}_{\text{ml}}}}\sqrt{P{{\xi }_{u,i}}}{{\hat{h}}_{u,i}}{{q}_{i}}} \\
    & + \sum\limits_{i\in {{\mathcal{S}}_{\text{CS}}}}{\sqrt{{{G}_{\text{ml}}}}\left( \sum\limits_{{{t}_{1}}=1}^{v-1}{\sqrt{P{{\xi }_{{{t}_{1}},i}}}}+\varpi \sum\limits_{{{t}_{2}}=v+1}^{K}{\sqrt{P{{\xi }_{{{t}_{2}},i}}}} \right){{\hat{h}}_{u,i}}{{q}_{i}}} \\
    & + \sqrt{{{G}_{\text{sl}}}}\sum\limits_{j\in {{\mathcal{I}}_{\text{IS}}}}{\sqrt{P}{{\hat{h}}_{u,j}}{{q}_{j}}}+{{\omega }_{u}}.
\end{aligned}
\end{equation}
Then, the SINR at UT$_u$ of the message intended for UT$_v$ is written as
\begin{equation}
\begin{aligned}
    \text{SINR}_{u}^{v}
    & = \frac{\sum\limits_{i\in {{\mathcal{S}}_{\text{CS}}}}{{{\left| {{h}_{u,i}} \right|}^{2}}{{\xi }_{v,i}}}}{\sum\limits_{i\in {{\mathcal{S}}_{\text{CS}}}}{I_{i}^{\text{intra}}}+\sum\limits_{j\in {{\mathcal{I}}_{\text{IS}}}}{I_{j}^{\text{inter}}}+{{{\bar{\sigma }}}^{2}}}.
\end{aligned}
\label{Formula_CoMS_NOMA_SINR}
\end{equation}
For simplicity of presentation in the following sections, we also denote the aggregated intra- and inter-satellite interference as ${{I}^{\text{intra}}}=\sum\limits_{i\in {{\mathcal{S}}_{\text{CS}}}}{I_{i}^{\text{intra}}}$ and ${{I}^{\text{inter}}}=\sum\limits_{j\in {{\mathcal{I}}_{\text{IS}}}}{I_{j}^{\text{inter}}}$, respectively.

Herein, NC-JT is applied to realize JT, where cooperative LEO satellites transmit identical data to the NOMA UT without prior phase mismatch correction or tight synchronization \cite{sun2024tochastic}. Consequently, useful signals are combined at the receiver, yielding a power boost known as cyclic delay diversity (CDD) \cite{sun2024on}. 
While strict phase alignment is not performed in NC-JT, independently received signal components from cooperative satellites still contribute to aggregated received power in a non-coherent manner. Therefore, NC-JT can still achieve cooperative diversity and signal enhancement gains while avoiding high synchronization overhead by coherent joint transmission in satellite networks.
More details on the CDD mechanism are available in \cite[Appendix A]{tanbourgi2014tractable}.
Following similar steps as in \cite{tanbourgi2014tractable}, the SINR expression in \eqref{Formula_CoMS_NOMA_SINR} is therefore obtained.

\begin{remark}
(Practical Network-Level Coordination): In practical implementations, CoMS transmission can be enabled through network-level coordination. A central controller, which may be implemented at a ground gateway or higher-layer satellite, e.g., LEO/MEO/GEO satellite, first collects network state information, including satellite visibility, user locations, traffic demands, and channel quality indicators.
On this basis, the controller performs cooperative satellite selection, user scheduling, NOMA grouping, and power allocation design. The resulting coordination information is then distributed to participating LEO satellites via feeder links, inter-satellite links (ISLs) \cite{song2025neuromorphic}, or optical ISLs \cite{shang2025channel} with low-latency forwarding capability. 
To support cooperative transmission, synchronization and scheduling information are periodically exchanged among cooperative satellites. The considered NC-JT framework relaxes the synchronization requirement by allowing satellites to transmit identical data streams without strict phase alignment. Consequently, only scheduling information, user grouping decisions, transmission timing, and power allocation coefficients need to be coordinated among satellites, which significantly reduces signaling complexity and improves implementation feasibility for LEO satellite systems.
\end{remark}

\section{Angle and Distance Distributions}
Given connections between multiple cooperative satellites and multiple served UTs, expressions for coverage probability and spectral efficiency involve the statistical distributions of relative distances and angles.
In addition, some other relevant results have been readily obtained in  \cite{shang2024multi,li2026downlink}, so they are provided directly in the following Subsections C and D.

\subsection{Zenith Angle of Random Satellite}
From point $[0,0,R_{\text{E}}]$, the zenith angle $\varphi_{\text{z}}$ and service region angle $\eta$ are marked as shown in Fig. \ref{fig:Angle}, and the dome angle $\varphi_{\text{d}}$ is measured from the center of the Earth pointing to the zenith and dome edge \cite{al2021analytic}.

\captionsetup{font={scriptsize}}
\begin{figure}[tp]
\begin{center}
\setlength{\abovecaptionskip}{+0.2cm}
\setlength{\belowcaptionskip}{-0.2cm}
\centering
  \includegraphics[width=3.0in, height=2.3in]{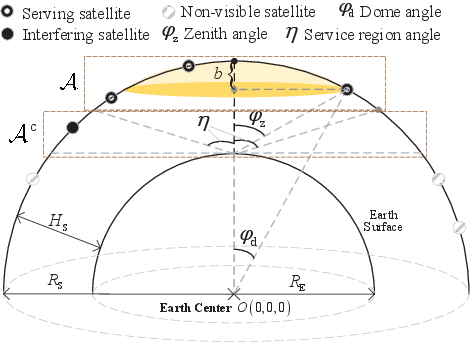}
\renewcommand\figurename{Fig.}
\caption{\scriptsize An illustration of angular relationships and satellite categorization. Based on their locations, satellites can be divided into serving satellites, interfering satellites, and remaining non-visible satellites. $\eta$ is the service region angle, while $\varphi_{\text{d}}$ and $\varphi_{\text{z}}$ are the dome angle and zenith angle, respectively.}
\label{fig:Angle}
\end{center}
\vspace{-5mm}
\end{figure}

\captionsetup{font={scriptsize}}
\begin{figure}[tp]
\begin{center}
\setlength{\abovecaptionskip}{+0.2cm}
\setlength{\belowcaptionskip}{-0.0cm}
\centering
  \includegraphics[width=3.0in, height=2.3in]{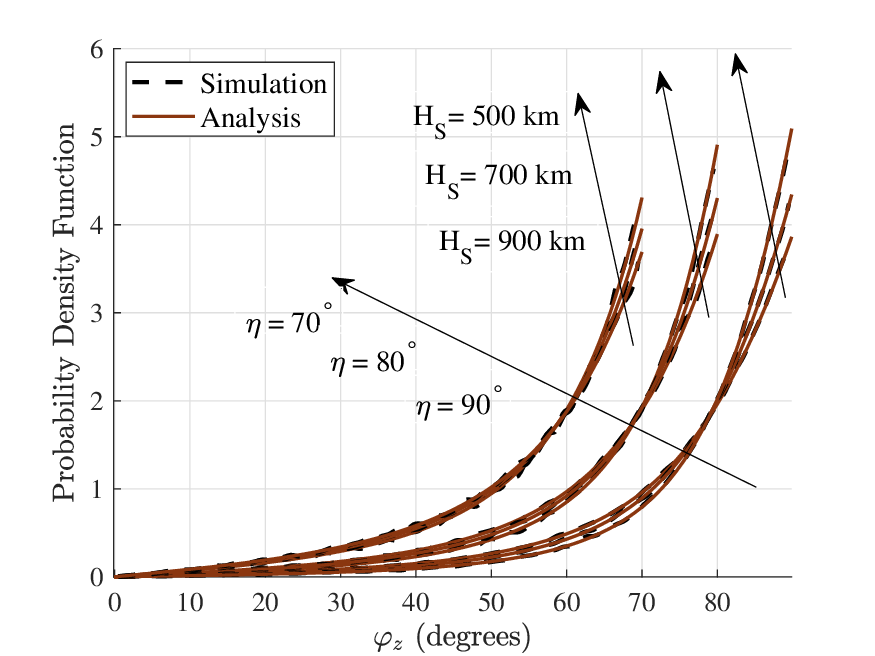}
\renewcommand\figurename{Fig.}
\caption{PDF of $\varphi_{\text{z}}$ under different service region angles $\eta$ and satellite altitudes $H_{\text{S}}$.}
\label{Exp_fig3}
\end{center}
\vspace{-8mm}
\end{figure}

\begin{proposition}
The PDF of the zenith angle $\varphi_z$ from a random satellite is written as 
\begin{equation}
    {{f}_{{{\varphi }_{\text{z}}}}}\left( {{\varphi }_{\text{z}}} \right)
    = \frac{{{\sin }^{3}}\left( \psi \left( {{\varphi }_{\text{z}}} \right) \right)\left[ 1+\frac{\kappa \cot {{\varphi }_{\text{z}}}}{\sqrt{1+{{\cot }^{2}}{{\varphi }_{\text{z}}}-{{\kappa }^{2}}}} \right]}{\left[ 1-\cos \left( \psi \left( \eta  \right) \right) \right]\left( 1-{{\kappa }^{2}} \right){{\sin }^{2}}{{\varphi }_{\text{z}}}},
\label{Formula:PDF_varphi_z}
\end{equation}
where $\psi \left( x \right)={{\cot }^{-1}}\left( \frac{\cot x+\sqrt{{{\kappa }^{2}}\left( 1+{{\cot }^{2}}x-{{\kappa }^{2}} \right)}}{1-{{\kappa }^{2}}} \right)$, and $\kappa =\frac{{{R}_{\text{E}}}}{{{R}_{\text{E}}}+{{H}_{\text{S}}}}$.
\end{proposition}

\begin{proof}
    Please see Appendix A.
\end{proof}

Fig. \ref{Exp_fig3} compares the PDF of $\varphi_{\text{z}}$ under three example service region angles $\eta=70^{\circ},80^{\circ},90^{\circ}$ and three example satellite altitudes $H_{\text{S}}=500{\,}{\text{km}}, 700{\,}{\text{km}}, 900{\,}{\text{km}}$, respectively. Thus, the correctness of the PDF derived from $\varphi_{\text{z}}$ is verified.

\vspace{-4mm}
\subsection{Distance from Random Satellite to Random UT}
As shown in Fig. \ref{fig:User_dist}, $\theta$ is the angle between the line from a random satellite S to point O$_\text{T}$, i.e., $(0,0,R_{\text{E}})$, (line O$_\text{T}$-S) and the line from a random UT$_1$ to point O$_\text{T}$ (line O$_\text{T}$-UT$_1$). There exists a minimum $\theta_{\text{min}}$ and a maximum $\theta_{\text{max}}$ for $\theta$ based on their locations.
$\zeta$ is the angle between line O$_\text{T}$-UT$_1$ and line O$_\text{T}$-S projected onto the area.
In addition, $r_{\text{N}}$ is the distance from point O$_\text{T}$ to a random satellite S, and ${r}_{\text{SU}}$ is the distance between satellite S and UT$_1$.

\captionsetup{font={scriptsize}}
\begin{figure}[tp]
\begin{center}
\setlength{\abovecaptionskip}{+0.2cm}
\setlength{\belowcaptionskip}{-0.2cm}
\centering
  \includegraphics[width=3.2in, height=2.0in]{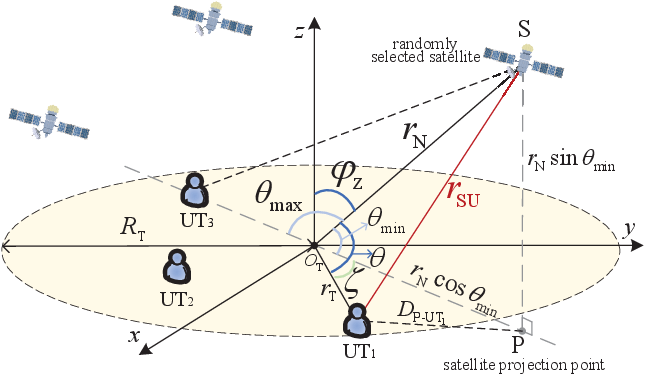}
\renewcommand\figurename{Fig.}
\caption{\scriptsize The angular and distance relationships, focusing on the typical serving area.}
\label{fig:User_dist}
\end{center}
\vspace{-3mm}
\end{figure}

\captionsetup{font={scriptsize}}
\begin{figure}[tp]
\begin{center}
\setlength{\abovecaptionskip}{+0.2cm}
\setlength{\belowcaptionskip}{-0.0cm}
\centering
  \includegraphics[width=3.2in, height=1.9in]{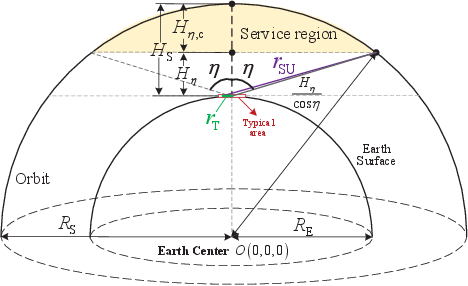}
\renewcommand\figurename{Fig.}
\caption{\scriptsize An illustration of angular and distance relationships, where $r_{\text{T}}$ is the distance from a random UT$_u$ to the center of typical area, $r_{\text{SU}}$ is the distance from a random satellite to a random UT$_u$, and $\eta$ is the service region angle.}
\label{fig:Spherical_cap}
\end{center}
\vspace{-8mm}
\end{figure}

\begin{lemma}
The relationships among $\theta$, ${{\varphi }_{\text{z}}}$, and $\zeta$ is represented as 
\begin{equation}
    \cos \theta =\sin {{\varphi }_{\text{z}}}\cos \zeta.
\end{equation}
\end{lemma}

\begin{proof}
Consider geometry in two triangles. In a triangle formed by a random satellite S, its projection point P, and UT$_1$, the distance between point P and UT$_1$ is given by
\begin{equation}
    {{D}_{\text{P-U}{{\text{T}}_{\text{1}}}}}=\sqrt{{{r}_{\text{SU}}}^{2}-{{\left( {{r}_{\text{N}}}\sin {{\theta }_{\min }} \right)}^{2}}}.
\end{equation}

In a triangle formed by point O$_\text{T}$, point P, and UT$_1$, we have the distance between point P and UT$_1$ given by
\begin{equation}
    {{D}_{\text{P-U}{{\text{T}}_{\text{1}}}}}=\sqrt{{{r}_{\text{T}}}^{2}+{{\left( {{r}_{\text{N}}}\cos {{\theta }_{\min }} \right)}^{2}}-2{{r}_{\text{T}}}{{r}_{\text{N}}}\cos {{\theta }_{\min }}\cos \zeta },
\end{equation}
where ${r}_{\text{T}}$ is the distance from point O$_\text{T}$ to a random UT$_1$, and $\zeta$ is marked in  Fig. \ref{fig:User_dist}. 
Combining the above equations and performing some mathematical manipulations yields the relationship above, which completes the proof.
\end{proof}

\begin{lemma}
The PDF of angle $\theta$ can be written as
\begin{equation}
    {{f}_{\theta }}\left( \theta  \right)=\sin \theta \int_{0}^{\eta }{\frac{{{f}_{{{\varphi }_{\text{z}}}}}\left( {{\varphi }_{\text{z}}} \right)}{\pi \sqrt{{{\sin }^{2}}{{\varphi }_{\text{z}}}-{{\cos }^{2}}\theta }}d{{\varphi }_{\text{z}}}},
\end{equation}
where ${{f}_{{{\varphi }_{\text{z}}}}}\left( {{\varphi }_{\text{z}}} \right)$ is shown in \eqref{Formula:PDF_varphi_z}.
\end{lemma}

\begin{proof}
    Please see Appendix B.
\end{proof}

Note that the intermediate results in the proof of Lemma 2 will be used for the proof of Proposition 2.

\begin{lemma}
The range of distances from a random serving satellite in the spherical cap to a random UT in the serving area is written as
\begin{equation}
    {{H}_{\text{S}}}\le {{r}_{\text{SU}}}\le \sqrt{{{r}_{\text{T}}}^{2}+{{\left( \frac{{{H}_{\eta }}}{\cos \eta } \right)}^{2}}-2{{r}_{\text{T}}}\frac{{{H}_{\eta }}}{\cos \eta } \sin \eta},
\label{Formula:range_r_SU}
\end{equation}
where ${{H}_{\eta }}=\left( \sqrt{{{R}_{\text{E}}}^{2}+\frac{{{R}_{\text{S}}}^{2}-{{R}_{\text{E}}}^{2}}{{{\cos }^{2}}\eta }}-{{R}_{\text{E}}} \right){{\cos }^{2}}\eta$.
\end{lemma}

\begin{proof}
    Please see Appendix C.
\end{proof}

\begin{proposition}
The conditional PDF of the distance from a random serving satellite to a random UT in the serving area is written as 
\begin{equation}
\begin{aligned}
    {{f}_{{{R}_{\text{SU}}}}}\left( {{r}_{\text{SU}}};{{\varphi }_{\text{z}}},{{r}_{\text{T}}},{{r}_{\text{N}}} \right)=\frac{{{r}_{\text{SU}}}}{2\pi {{r}_{\text{N}}}{{r}_{\text{T}}}\sin {{\varphi }_{\text{z}}}\sqrt{1-{{\left( \frac{{{r}_{\text{N}}}^{2}+{{r}_{\text{T}}}^{2}-{{r}_{\text{SU}}}^{2}}{2{{r}_{\text{N}}}{{r}_{\text{T}}}\sin {{\varphi }_{\text{z}}}} \right)}^{2}}}}.
    \label{Formula:PDF_varphi_R_SU}
\end{aligned}
\end{equation}
\end{proposition}

\begin{proof}
    Please see Appendix D.
\end{proof}

\subsection{Distance of Ordered UT links}
The PDF of the distance ${r}_{\text{T}}$ from the ordered UT$_u$ to the center O$_\text{T}$ of the serving area is given by \cite{shang2024multi}
\begin{equation}
\begin{aligned}
    {{f}_{{{R}_{\text{T}}},u}}\left( {{r}_{\text{T}}} \right) \nonumber 
     & = K \left( \begin{matrix}
       K-1  \\
       i-1  \\
    \end{matrix} \right){{\left( \frac{{{r}_{\text{T}}}^{2}}{{{R}_{\text{T}}}^{2}} \right)}^{u-1}}{{\left( 1-\frac{{{r}_{\text{T}}}^{2}}{{{R}_{\text{T}}}^{2}} \right)}^{K-u}}\frac{2{{r}_{\text{T}}}}{{{R}_{\text{T}}}^{2}} ,
\end{aligned}
\end{equation}
where $\left( \begin{matrix} p\\ q\\ \end{matrix} \right)=\frac{p!}{q!\left( p-q \right)!}$.

\subsection{Distance to Random Satellite}
The cumulative distribution function (CDF) of the distance ${r}_{\text{N}}$ from the center of the serving area to a random satellite is given by \cite{li2026downlink}
\begin{equation}
    {{F}_{{{R}_{\text{N}}}}}\left( {{r}_{\text{N}}} \right)
    =\frac{{{r}_{\text{N}}}^{2}-{{\left( {{R}_{\text{S}}}-{{R}_{\text{E}}} \right)}^{2}}}{2{{R}_{\text{E}}}\left( {{R}_{\text{S}}}-{{R}_{\text{E}}}{{\sin }^{2}}\eta -\sqrt{{{R}_{\text{S}}}^{2}-{{R}_{\text{E}}}^{2}{{\sin }^{2}}\eta }\cos \eta  \right)},
\end{equation}
for $\left( {{R}_{\text{S}}}-{{R}_{\text{E}}} \right)<{{r}_{\text{N}}}<\sqrt{{{R}_{\text{S}}}^{2}-{{R}_{\text{E}}}^{2}{{\sin }^{2}}\eta }-{{R}_{\text{E}}}\cos \eta$. 
Then, the corresponding PDF of ${r}_{\text{N}}$ is given by
\begin{equation}
    {{f}_{{{R}_{\text{N}}}}}\left( {{r}_{\text{N}}} \right)
    =\frac{{{r}_{\text{N}}}}{{{R}_{\text{E}}}\left( {{R}_{\text{S}}}-{{R}_{\text{E}}}{{\sin }^{2}}\eta -\sqrt{{{R}_{\text{S}}}^{2}-{{R}_{\text{E}}}^{2}{{\sin }^{2}}\eta }\cos \eta  \right)}.
\end{equation}

\section{Performance Analysis}
With derived angle and distance distributions, this section presents a performance analysis framework for the CoMS-NOMA network. We first analyze the inter- and intra-satellite interference terms together with the coverage conditions under SIC. Then, the statistical distribution of the combined desired signal power is derived. Finally, analytical expressions for coverage probability and spectral efficiency are obtained.

\subsection{Coverage, Interference Analysis, and SIC Ordering}

The interference includes inter-satellite interference from non-serving, interfering satellites and NOMA intra-satellite interference from serving, cooperative satellites.
In the following, we first define the coverage event in the CoMS-NOMA network. Then, we derive the Laplace transform of the inter-satellite interference term, denoted by ${{I}^{\text{inter}}} = \sum\limits_{j\in {{\mathcal{I}}_{\text{IS}}}} {I_{j}^{\text{inter}}} = \sum\limits_{j\in {{\mathcal{I}}_{\text{IS}}}} \frac{{{G}_{\text{sl}}}}{{{G}_{\text{ml}}}}{{\left| {{h}_{u,j}} \right|}^{2}}$, and analyze the intra-satellite interference term ${{I}^{\text{intra}}}=\sum\limits_{i\in {{\mathcal{S}}_{\text{CS}}}}{I_{i}^{\text{intra}}}$ using mathematical manipulations.

\begin{definition}
(Coverage in CoMS-NOMA Network): 
Since multiple UTs are concurrently served by multiple cooperative satellites, the coverage event must be satisfied by each UT when it is served by each cooperative satellite. 
Consider a uniform SINR threshold $\gamma$ for all intended UTs. If the SINR at UT$_u$ of the message intended for all other stronger UTs is higher than $\gamma$, i.e., $\bigcap\limits_{v=u}^{K}{\left\{ \text{SINR}_{u}^{v}>\gamma \right\}}$, 
UT$_u$ is assumed to be in coverage; for simplicity, we denote 
$\Lambda_{u}^{v}=\left\{ \text{SINR}_{u}^{v}>\gamma  \right\}$. 
Then, the coverage event at UT$_u$ is defined as
\begin{equation}
    {{\Lambda }_{u}}=\left\{ \bigcap\limits_{v=u}^{K}{\Lambda _{u}^{v}} \right\}=\left\{ \bigcap\limits_{v=u}^{K}{\left\{ \text{SINR}_{u}^{v}>\gamma  \right\}} \right\}.
\end{equation}
\end{definition}
The coverage probability at UT$_u$ can be expressed as
\begin{align}
    & \mathbb{P}\left( {{\Lambda }_{u}} \right)  
    =\mathbb{P}\left( \bigcap\limits_{v=u}^{K}{\left\{ \text{SINR}_{u}^{v}>\gamma  \right\}} \right) \nonumber \\ 
    & =\mathbb{P}\left( \bigcap\limits_{v=u}^{K}{\left\{ \frac{\sum\limits_{i\in {{\mathcal{S}}_{u}}}{{{\left| {{h}_{u,i}} \right|}^{2}}{{\xi }_{v,i}}}}{\sum\limits_{i\in {{\mathcal{S}}_{\text{CS}}}}{I_{i}^{\text{intra}}}+\sum\limits_{j\in {{\mathcal{I}}_{\text{IS}}}}{I_{j}^{\text{inter}}}+{{{\bar{\sigma }}}^{2}}}>\gamma  \right\}} \right).
\label{Formula:P_Lambda_u_1}
\end{align}
%
Herein, each cooperating satellite is assumed to employ the same power allocation ratio across the NOMA UTs, i.e., the relative power ratios assigned to different UTs are identical among all cooperative satellites. 
This assumption is mainly adopted for analytical tractability and to avoid excessive coordination overhead among cooperative satellites\footnote{Such implementation also reflects a practical low-complexity implementation strategy, where cooperative satellites follow a unified NOMA transmission configuration while still capturing the essential characteristics of CoMS-NOMA transmission.}.
Then, by transposition, the coverage probability at UT$_u$ is expressed as
\begin{align}
    & \mathbb{P}\left( {{\Lambda }_{u}} \right) \nonumber \\ 
    & =\mathbb{P}\left( \bigcap\limits_{v=u}^{K} \left\{ \sum\limits_{i\in {{\mathcal{S}}_{\text{CS}}}}{{{\left| {{h}_{u,i}} \right|}^{2}}{{\xi }_{v,i}}}  \right. \right. \nonumber \\
    & {\qquad \qquad \qquad \qquad } \left. \left. > \gamma \left[ \sum\limits_{i\in {{\mathcal{S}}_{\text{CS}}}}{I_{i}^{\text{intra}}}+\sum\limits_{j\in {{\mathcal{I}}_{\text{IS}}}}{I_{j}^{\text{inter}}}+{{{\bar{\sigma }}}^{2}} \right] \right\} \right) \nonumber \\ 
    & = \mathbb{P}\left( \bigcap\limits_{v=u}^{K} \left\{ \sum\limits_{i\in {{\mathcal{S}}_{\text{CS}}}}{ {{\left| {{h}_{u,i}} \right|}^{2}} } >  \left[ \sum\limits_{j\in {{\mathcal{I}}_{\text{IS}}}}{I_{j}^{\text{inter}}}+{{{\bar{\sigma }}}^{2}} \right] \frac{\gamma }{\tilde{\xi}_{v}} \right\} \right),
\label{Formula:P_Lambda_u_2}
\end{align}
where $\tilde{\xi}_{v} = {{\xi}_{v}}-\gamma \left( \sum\limits_{v<u}{{{\xi }_{v}}}+\varpi \sum\limits_{v>u}{{{\xi }_{v}}} \right)$.
This reformulation converts the original SINR constraints into an equivalent desired-signal power condition, which facilitates subsequent interference and coverage analysis.

Take ${{Q}_{u}} = \underset{u\le v\le K}{\mathop{\max }}\,\frac{\gamma }{ \tilde{\xi}_{v} }$, the coverage probability at UT$_u$ is expressed as
\begin{equation}
    \mathbb{P}\left( {{\Lambda }_{u}} \right) = \mathbb{P}\left( \sum\limits_{i\in {{\mathcal{S}}_{\text{CS}}}}{{ {\left| {{h}_{u,i}} \right|}^{2} } }>\left[ \sum\limits_{j\in {{\mathcal{I}}_{\text{IS}}}}{I_{j}^{\text{inter}}}+{{{\bar{\sigma }}}^{2}} \right] {{Q}_{u}} \right).
\label{Formula:P_Lambda_u_3}
\end{equation}
Denote $r_{\text{N}}$ as the distance between point O$_\text{T}$ and any of the interfering satellites. For analytical simplicity, when calculating large-scale inter-satellite interference, we assume that UTs are located at O$_\text{T}$, and use the approximate mean of the distance $r_{\text{N}}$.
Then, considering the uniform distribution of UTs in the typical serving area, the distance between a UT and any one of the interfering satellites, denoted by $r_{\text{IU}}$, can be approximated by $\mathbb{E}\left\{r_{\text{IU}}|r_{\text{N}} \right\}\approx r_{\text{N}}$. 
A similar approximation method was also adopted in TN \cite{ali2019downlink,liang2019non} for analytical tractability.
 
This approximation is reasonable because the serving area radius is relatively small compared with the satellite-Earth distance in LEO systems, such that the distance variation among different UT locations has a limited impact on the dominant large-scale interference characteristics.

\begin{lemma}
(Inter-satellite interference): The Laplace transform of the inter-satellite interference term $I_{j}^{\text{inter}}$ is given by
\begin{align}
    & {{\mathcal{L}}_{ {{I}^{\text{inter}}} = \sum\limits_{j\in {{\mathcal{I}}_{\text{IS}}}} {I_{j}^{\text{inter}}} }}\left( s \right) \nonumber \\
    & = \exp \left( -{{\lambda }_{\text{S}}}\pi \frac{{{R}_{\text{S}}}}{{{R}_{\text{E}}}}{{\left( \frac{{{G}_{\text{sl}}}}{{{G}_{\text{ml}}}}\frac{s}{\kappa } \right)}^{\frac{2}{\alpha }}} \right. \nonumber \\
    & {\quad\quad\quad\quad} \cdot \left. \int_{{{\left( \frac{{{G}_{\text{sl}}}}{{{G}_{\text{ml}}}}\frac{s}{\kappa } \right)}^{-\frac{2}{\alpha }}}{{r}_{\text{N}}}^{2}}^{{{\left( \frac{{{G}_{\text{sl}}}}{{{G}_{\text{ml}}}}\frac{s}{\kappa } \right)}^{-\frac{2}{\alpha }}}{{R}_{\text{max} }}^{2}}{\left[ 1-\frac{1}{{{\left( 1+{{u}^{-\frac{\alpha }{2}}} \right)}^{\kappa }}} \right]du} \right),
\end{align}
where 
$R_{\text{max}} = \sqrt{{R_\text{S}}^2 - {R_\text{E}}^2}$.
\end{lemma}
\begin{proof}
As shown in Figs. \ref{fig:User_dist} and \ref{fig:Spherical_cap}, the distance range from point O$_\text{T}$ to a random satellite S is $R_{\text{min}} \le r_{\text{N}} \le R_{\text{max}}$, where $R_{\text{min}} = R_{\text{Ser,min}} = H_{\text{S}}$, and $R_{\text{max}} = \sqrt{{R_\text{S}}^2 - {R_\text{E}}^2}$.
By excluding the approximate distance range, i.e., $R_{\text{Ser,min}} \le r_{\text{N}} \le R_{\text{Ser,max}}$, of the desired cooperative signal, the approximate distance range of the aggregated interference signal can be written as $R_{\text{Ser,max}} \le r_{\text{N}} \le R_{\max}$.

The service region on the sphere has been denoted by $\mathcal{A}$ in Appendix A, and here we further denote its complementary area, i.e., the rest of the spherical cap above the typical serving area, by $\mathcal{A}^{\text{c}}$, as shown in Fig. \ref{fig:Angle}.
Then, the Laplace transform of the aggregated interference power can be calculated as
\begin{align}
    & {{\mathcal{L}}_{ {{I}^{\text{inter}}} = \sum\limits_{j\in {{\mathcal{I}}_{\text{IS}}}} {I_{j}^{\text{inter}}} }}\left( s \right) \nonumber \\
    &=\mathbb{E}\left[ {{e}^{-sI_{j}^{\text{inter}}}}\left| \left\| {\mathbf{s}_{j}}-{\mathbf{x}_{u}} \right\|={{r}_{\text{N}}},j\in {{\mathcal{I}}_{\text{IS}}} \right. \right] \nonumber \\ 
    & \overset{\left( a \right)}{\mathop{=}} \,\exp \left( -{{\lambda }_{\text{S}}}\int_{v\in \mathcal{A}_{{{r}_{_{\text{N}}}}}^{\text{c}}}{\left( 1-\mathbb{E}\left[ {{e}^{-s\frac{{{G}_{\text{sl}}}}{{{G}_{\text{ml}}}}{{\left| {{h}_{j,u}} \right|}^{2}}{{v}^{-\alpha }}}} \right] \right)dv} \right) \nonumber \\ 
    & \overset{\left( b \right)}{\mathop{=}}\,\exp \left( -{{\lambda }_{\text{S}}}\int_{v\in \mathcal{A}_{{{r}_{_{\text{N}}}}}^{\text{c}}}{\left[ 1-\frac{1}{{{\left( 1+\frac{{{G}_{\text{sl}}}}{{{G}_{\text{ml}}}}\frac{s{{v}^{-\alpha }}}{\kappa } \right)}^{\kappa }}} \right]dv} \right) \nonumber \\ 
    & \overset{\left( c \right)}{\mathop{=}}\,\exp \left( -2\pi {{\lambda }_{\text{S}}}\frac{{{R}_{\text{S}}}}{{{R}_{\text{E}}}} \int_{{r}_{\text{N}}}^{{{R}_{\max }}}{\left[ 1-\frac{1}{{{\left( 1+\frac{{{G}_{\text{sl}}}}{{{G}_{\text{ml}}}}\frac{s{{v}^{-\alpha }}}{\kappa } \right)}^{\kappa }}} \right]dv} \right) \nonumber \\ 
    & \overset{\left( d \right)}{\mathop{=}}\,\exp \left( -{{\lambda }_{\text{S}}}\pi \frac{{{R}_{\text{S}}}}{{{R}_{\text{E}}}}{{\left( \frac{{{G}_{\text{sl}}}}{{{G}_{\text{ml}}}}\frac{s}{\kappa } \right)}^{\frac{2}{\alpha }}} \right. \nonumber \\
    & {\quad\quad\quad\quad} \cdot \left. \int_{{{\left( \frac{{{G}_{\text{sl}}}}{{{G}_{\text{ml}}}}\frac{s}{\kappa } \right)}^{-\frac{2}{\alpha }}}{{r}_{\text{N}}}^{2}}^{{{\left( \frac{{{G}_{\text{sl}}}}{{{G}_{\text{ml}}}}\frac{s}{\kappa } \right)}^{-\frac{2}{\alpha }}}{{R}_{\max }}^{2}}{\left[ 1-\frac{1}{{{\left( 1+{{u}^{-\frac{\alpha }{2}}} \right)}^{\kappa }}} \right]du} \right),
\end{align}
where (a) is because of the probability generating functional (PGFL) of PPP \cite{park2022tractable,li2026coverage}, (b) is due to the distribution properties of Nakagami-$m$ distribution, (c) follows from the derivative of $\frac{\partial \left| {{\mathcal{A}}_{{{r}_{\text{N}}}}} \right|}{\partial {{r}_{\text{N}}}}=2\frac{{{R}_{\text{S}}}}{{{R}_{\text{E}}}}\pi {{r}_{\text{N}}}$, and (d) comes from the change of variable, i.e., $u={{\left( \frac{{{G}_{\text{sl}}}}{{{G}_{\text{ml}}}}\frac{s}{\kappa } \right)}^{-\frac{2}{\alpha }}}{{v}^{2}}$ followed by $du=2v{{\left( \frac{{{G}_{\text{sl}}}}{{{G}_{\text{ml}}}}\frac{s}{\kappa } \right)}^{-\frac{2}{\alpha }}}$, which completes the proof.
\end{proof}
The above Laplace transform is performed to characterize the aggregated inter-satellite interference and to enable tractable coverage probability analysis under the considered PPP-based satellite distributions.

\begin{remark}
    (Intra-satellite interference): According to \eqref{Formula:P_Lambda_u_2}, the impact of intra-satellite interference can be regarded as a reduction in the effective power allocation coefficient of UT$_v$, and the corresponding $\tilde{\xi}_{v}$ is not related to the transmission rate of the message to be decoded. 
\end{remark}

To ensure successful operation for all UTs decoding the received message under the NOMA mechanism, the requirement $\tilde{\xi}_{v} > 0$ must be satisfied; otherwise, the transmission rate will be too high for the given power allocation, and thus the SIC mechanism cannot be successfully operated. This requirement is also referred to as the ``NOMA necessary condition'' \cite{ali2019downlink}.

\begin{remark}
(SIC Ordering in CoMS-NOMA Network): UTs are served in a uniform order by cooperative satellites. Instead of computing the exact mean signal power (MSP) by aggregating distances from each UT to all cooperative satellites, which requires per-satellite distance estimation and significantly increases implementation overhead, we approximate the ordering by the UTs’ distances to the center of the serving area.
This center-distance-based ordering serves as an effective surrogate for MSP ordering\footnote{Instantaneous SINR (ISINR)-based ordering generally achieves slightly better performance than MSP ordering. However, as demonstrated in \cite{li2026coverage}, both methods yield comparable overall performance, while MSP ordering offers significantly lower complexity since it does not require UTs to provide feedback or coordinate ISINR information among cooperative satellites.} because 
UTs closer to the service-area center generally maintain shorter average distances to multiple cooperative satellites simultaneously, thereby experiencing stronger aggregated mean signal power under cooperative transmission.
\end{remark}

From an implementation perspective, the MSP ordering can be used as the SIC ordering, with UTs prioritized by their distance to the service area center. This ordering scheme only requires UT location information, which can be obtained via existing positioning mechanisms with minimal signaling overhead. The computational complexity mainly involves simple distance calculations and sorting operations, which scale linearly with the number of users and are therefore computationally tractable. 
In contrast, CSI-based or ISINR ordering would require frequent channel estimation and feedback in dynamic LEO environments, potentially incurring significantly higher signaling overhead.

\subsection{Combined Signal Power Distribution}
The desired signal power from the $i$-th cooperative satellite to UT$_u$ involves a random variable (RV) ${{V}_{i}} = {\left| {{h}_{u,i}} \right|}^{2} = \frac{{{\left| {{g}_{u,i}} \right|}^{2}}}{{{r}_{u,i}}^{\alpha }}$, where ${{g}_{u,i}}$ and ${{r}_{u,i}}$ are the channel and distance between the $i$-th cooperative satellite and UT$_u$, respectively. The corresponding CDF can be expressed as ${{F}_{{{V}_{i}}}}\left( v \right)=\text{GamCDF}\left( {{m}_{u,i}},{{r}_{u,i}} \right)$. For $N_{\text{cs}}$ cooperative satellites, we denote ${{V}_{\text{sum}}}=\sum_{i=1}^{\left| {{\mathcal{S}}_{u}} \right|=N_{\text{cs}}} V_i={{V}_{1}}+{{V}_{2}}+\ldots +{{V}_{\left| {{\mathcal{S}}_{u}} \right|}}$. 

Subsequently, to obtain the CDF of ${V}_{\text{sum}}$, its moment generating function (MGF) is expressed as
\begin{align}
& {{\mathcal{M}}_{{{V}_{\text{sum}}}}}\left( s,r \right) \nonumber \\
& ={{\mathbb{E}}_{{{V}_{1}},{{V}_{2}},\ldots ,{{V}_{\left| {{\mathcal{S}}_{\text{CS}}} \right|}}}}\left[ \exp \left( s\left( {{V}_{1}}+{{V}_{2}}+\ldots +{{V}_{\left| {{\mathcal{S}}_{\text{CS}}} \right|}} \right) \right) \right] \nonumber \\ 
& \begin{aligned}
= & {{\left( 1-\frac{{{r}_{u,1}}^{-\alpha }}{{{m}_{u,1}}}s \right)}^{-{{m}_{u,1}}}}\cdot 
{{\left( 1-\frac{{{r}_{u,2}}^{-\alpha }}{{{m}_{u,2}}}s \right)}^{-{{m}_{u,2}}}} \\
& \cdots 
{{\left( 1-\frac{{{r}_{u,\left| {{\mathcal{S}}_{\text{CS}}} \right|}}^{-\alpha }}{{{m}_{u,\left| {{\mathcal{S}}_{\text{CS}}} \right|}}}s \right)}^{-{{m}_{u,\left| {{\mathcal{S}}_{\text{CS}}} \right|}}}}. \end{aligned}
\end{align}
Here, we drop the index $u$ for brevity. Using a partial fraction expansion, ${{\mathcal{M}}_{{{V}_{\text{sum}}}}}$ can be rewritten as
\begin{equation}
    {{\mathcal{M}}_{{{V}_{\text{sum}}}}}\left( s,{{r}_{i}} \right)={{\sum\limits_{i=1}^{{{N}_{\text{CS}}}}{\sum\limits_{l=1}^{{{m}_{i}}}{{{c}_{i,l}}\left( s-\frac{{{m}_{i}}}{{{r}_{i}}^{-\alpha }} \right)}}}^{-l}},
\end{equation}
where ${{c}_{i,l}}$ is calculated from the high-order derivative as
\begin{equation}
    {{c}_{i,l}}=\frac{1}{\left( {{m}_{i}}-l \right)!}{{\left( -{{s}_{i}} \right)}^{{{m}_{i}}}}\frac{{{\partial }^{{{m}_{i}}-l}}}{\partial {{s}^{{{m}_{i}}-l}}}{{\left. \left[ \prod\limits_{p\ne i}^{{{N}_{\text{CS}}}}{{{\left( 1-{{b}_{p}}s \right)}^{{{m}_{i}}}}} \right] \right|}_{s={{s}_{i}}}},
\end{equation}
and denote ${{b}_{p}}=\frac{1}{{{s}_{p}}}$. By applying the multivariate Leibniz rule, which is 
\begin{equation}
    \frac{{{\partial }^{n}}}{\partial {{s}^{n}}}{{\left( 1-bs \right)}^{-m}}={{b}^{n}}{{\left( m \right)}_{n}}{{\left( 1-bs \right)}^{-m-n}},
\end{equation}
where ${{\left( m \right)}_{n}}=\frac{\Gamma \left( m+n \right)}{\Gamma \left( m \right)}$ is the rising factorial. Using the standard coefficient formula at a pole of order ${{m}_{i}}$, the finite multi-index sum in ${{c}_{i,l}}$ can be further calculated as
\begin{align}
&{{c}_{i,l}} 
    =  {{\left( -\frac{1}{{{b}_{i}}} \right)}^{{{m}_{i}}}} \nonumber \\
    & {\qquad} \cdot  \sum\limits_{\begin{smallmatrix} 
 {{\left( {{n}_{p}} \right)}_{p\ne i}}\in \mathbb{Z}\ge 0 \\ 
 \sum\nolimits_{p\ne i}{{{n}_{p}}={{m}_{i}}-l} 
\end{smallmatrix}} {\prod\limits_{p\ne i}^{{{N}_{\text{CS}}}}{\left[ \left( \begin{matrix}
   {{m}_{p}}+{{n}_{p}}-1  \\
   {{n}_{p}}  \\
\end{matrix} \right)\frac{b_{k}^{nk}}{{{\left( 1-\frac{{{b}_{p}}}{{{b}_{i}}} \right)}^{{{m}_{p}}+{{n}_{p}}}}} \right]}}. 
\end{align}
By inserting ${c}_{i,l}$, the CDF of ${V}_{\text{sum}}$ can be expressed as
\begin{equation}
    {{F}_{{{V}_{\text{sum}}}}}\left( v \right)=\sum\limits_{i=1}^{{{N}_{\text{CS}}}}{\sum\limits_{l=1}^{{{m}_{i}}}{{{c}_{i,l}}\left( -\frac{{{r}_{i}}^{-\alpha }}{{{m}_{i}}} \right)}}\text{GamCDF}\left( v,l,{{m}_{i}},{{r}_{i}} \right).
\end{equation}
Considering the square-law relationship \cite{nakagami1960the,holm2004sum} between the Nakagami-$m$ distribution and the Gamma distribution as previously mentioned, and the uniform signal fading conditions, $m = m_{u,i} = \kappa$. Then, the CDF of ${V}_{\text{sum}}$ is given by
\begin{equation}
    {{F}_{{{V}_{\text{sum}}}}}\left( v \right)=\sum\limits_{i=1}^{{{N}_{\text{CS}}}}{\sum\limits_{l=1}^{{ \kappa }}{{{c}_{i,l}}\left( -\frac{{{r}_{i}}^{-\alpha }}{{ \kappa }} \right)}}\text{GamCDF}\left( v,l,{\kappa},{{r}_{i}} \right).
\end{equation}

For numerical analysis and the relationship between the cooperative satellite number $N_{\text{CS}}$ and service region angle $\eta$, a lower bound is applied to ensure that the resulting satellite number remains an integer. 
The analytical cooperative satellite number is calculated as $N_{\text{CS}}^{\text{ana}}=2\pi {{R}_{\text{S}}}{{H}_{\eta ,\text{c}}}{{\lambda }_{\text{S}}}$, where $2\pi {{R}_{\text{S}}}{{H}_{\eta ,\text{c}}}$ is the area in which the cooperative satellite is located, and ${{H}_{\eta ,\text{c}}}$ is in \eqref{Formula:H_eta_c} of Appendix C. 
Then, we have ${{N}_{\text{CS}}}=\left\lfloor N_{\text{CS}}^{\text{ana}} \right\rfloor $.

\subsection{Coverage Probability and Spectral Efficiency}

\begin{theorem}
(Coverage probability): In the CoMS-NOMA network, the coverage probability of UT$_u$ is expressed as
\begin{align}
    & \mathbb{P}\left( {{\Lambda }_{u}} \right) \nonumber \\
    & = \int_{0}^{{{R}_{\text{T}}}} \int_{{{R}_{\min }}}^{{{R}_{\max }}} \int_{0}^{\eta } \int_{{{R}_{\text{SU},\min }}}^{{{R}_{\text{SU},\max }}} \nonumber \\
    & {\quad} \left[ \sum\limits_{i=1}^{{{N}_{\text{CS}}}}{\sum\limits_{l=1}^{\kappa }{{{c}_{i,l}}{{\left( -\frac{{{r}_{i}}^{-\alpha }}{\kappa } \right)}^{l}}}}{{f}_{{{R}_{\text{SU}}}}}\left( {{r}_{i}};{{r}_{\text{N}}},{{r}_{\text{T}}},{{\varphi }_{\text{z}}} \right) \right. \nonumber \\
    & { \quad \quad} \left. 
    \cdot \sum \limits_{k=0}^{l-1} \frac{{{\left[ \left( I^{\text{inter}}+{{{\bar{\sigma }}}^{2}} \right){{Q}_{u}} \right]}^{k}}}{k!} {{\left( -1 \right)}^{k}} \frac{{{d}^{k}}\mathcal{L}_{{I}^{\text{inter}}} \left( s \right)}{d{{s}^{k}}} \right] d r_i \nonumber \\
    & {\qquad \qquad \quad \quad} 
    \cdot {{f}_{{{\varphi }_{\text{z}}}}}\left( {{\varphi }_{\text{z}}} \right) d{{\varphi }_{\text{z}}} {{f}_{{{R}_{\text{N}}}}} \left( {{r}_{i}} \right) d{{r}_{i}} {{f}_{{{R}_{\text{T}}},u}}\left( {{r}_{\text{T}}} \right) d{{r}_{\text{T}}},
\end{align}
\label{Formula_P_Lambda_u_4}
where the range of $r_i=r_{\text{SU}} \in \left[{{R}_{\text{SU},\min }}, {{R}_{\text{SU},\max }} \right]$ is provided in \eqref{Formula:range_r_SU} of Lemma 3 and its proof.
\end{theorem}
\begin{proof}
The coverage probability in \eqref{Formula:P_Lambda_u_3} can be rewritten as
\begin{align}
    & \mathbb{P}\left( {{\Lambda }_{u}} \right) \nonumber \\
    & = \mathbb{P}\left( \sum\limits_{i\in {{\mathcal{S}}_{\text{CS}}}}{\frac{{{\left| {{g}_{u,i}} \right|}^{2}}}{{{r}_{u,i}}^{\alpha }}} > \left( I^{\text{inter}}+{{{\bar{\sigma }}}^{2}} \right){{Q}_{u}} \right) \nonumber \\ 
    & \overset{\left( a \right)}{\mathop{=}} {{\mathbb{E}}_{ {{r}_{\text{T}}},{{r}_{\text{N}}},{\varphi }_{\text{z}} }} \left\{ \sum\limits_{i=1}^{{{N}_{\text{CS}}}}{\sum\limits_{l=1}^{{\kappa}}{{{c}_{i,l}}{{\left( -\frac{{{r}_{i}}^{-\alpha }}{{\kappa}} \right)}^{l}}}} {{f}_{{{R}_{\text{SU}}}}}\left( {{r}_{i}};{{r}_{\text{N}}},{{r}_{\text{T}}},{{\varphi }_{\text{z}}} \right) \right. \nonumber \\
    & {\qquad  \quad \quad \quad \quad} \left. 
    \cdot {{e}^{-\left( {{I}^{\text{inter}}}+{{{\bar{\sigma }}}^{2}} \right){{Q}_{u}}}}\sum\limits_{k=0}^{l-1}{\frac{{{\left[ \left( {{I}^{\text{inter}}}+{{{\bar{\sigma }}}^{2}} \right){{Q}_{u}} \right]}^{k}}}{k!}} \right\} \nonumber  \\
    & \overset{\left( b \right)}{\mathop{=}} {{\mathbb{E}}_{ {{r}_{\text{T}}},{{r}_{\text{N}}},{\varphi }_{\text{z}} }}\left\{  \sum\limits_{i=1}^{{{N}_{\text{CS}}}} {\sum\limits_{l=1}^{{\kappa}}{{{c}_{i,l}}{{\left( -\frac{{{r}_{i}}^{-\alpha }}{{\kappa}} \right)}^{l}}}} {{f}_{{{R}_{\text{SU}}}}}\left( {{r}_{i}};{{r}_{\text{N}}},{{r}_{\text{T}}},{{\varphi }_{\text{z}}} \right) \right. \nonumber \\
    & {\quad \quad \quad \quad} \left. 
    \cdot \sum\limits_{k=0}^{l-1}{\frac{{{\left[ \left( I^{\text{inter}}+{{{\bar{\sigma }}}^{2}} \right){{Q}_{u}} \right]}^{k}}}{k!}{{\left( -1 \right)}^{k}}}\frac{{{d}^{k}}\mathcal{L}_{{I}^{\text{inter}}} \left( s \right)}{d{{s}^{k}}}  \right\}  \nonumber \\
    & \overset{\left( c \right)}{\mathop{=}} \int_{0}^{{{R}_{\text{T}}}} \int_{{R_{\text{Ser,max}}}}^{{{R}_{\max }}} {\int_{0}^{\eta }} \int_{{{R}_{\text{SU},\min }}}^{{{R}_{\text{SU},\max }}} \nonumber \\
    & {\quad} \left[ \sum\limits_{i=1}^{{{N}_{\text{CS}}}}{\sum\limits_{l=1}^{\kappa }{{{c}_{i,l}}{{\left( -\frac{{{r}_{i}}^{-\alpha }}{\kappa } \right)}^{l}}}}{{f}_{{{R}_{\text{SU}}}}}\left( {{r}_{i}};{{r}_{\text{N}}},{{r}_{\text{T}}},{{\varphi }_{\text{z}}} \right) \right. \nonumber \\
    & {\qquad \quad } \left. 
    \cdot \sum \limits_{k=0}^{l-1} \frac{{{\left[ \left( I^{\text{inter}}+{{{\bar{\sigma }}}^{2}} \right){{Q}_{u}} \right]}^{k}}}{k!} {{\left( -1 \right)}^{k}} \frac{{{d}^{k}}\mathcal{L}_{{I}^{\text{inter}}} \left( s \right)}{d{{s}^{k}}} \right] d r_i \nonumber \\
    & {\qquad \qquad \quad \quad} 
    \cdot {{f}_{{{\varphi }_{\text{z}}}}}\left( {{\varphi }_{\text{z}}} \right) d{{\varphi }_{\text{z}}} {{f}_{{{R}_{\text{N}}}}} \left( {{r}_{\text{N}}} \right) d{{r}_{\text{N}}} {{f}_{{{R}_{\text{T}}},u}}\left( {{r}_{\text{T}}} \right) d{{r}_{\text{T}}},
\end{align}
where (a) follows from the CDF of the Gamma distribution, (b) is due to the derivative property of the Laplace transform, i.e., $\mathbb{E}\left[ {{X}^{k}}{{e}^{-sX}} \right]={{\left( -1 \right)}^{k}}\frac{{{d}^{k}}{{\mathcal{L}}_{X}}\left( s \right)}{d{{s}^{k}}}$ \cite{park2022tractable}, and (c) comes from the expectation over ${{r}_{\text{T}}}$, ${{r}_{\text{N}}}$, ${\varphi }_{\text{z}}$, which completes the proof.
\end{proof}
The obtained expression captures the joint impact of cooperative signal enhancement, inter-satellite interference, and SIC conditions on the coverage performance.
To evaluate the average coverage performance of the served UTs, their average coverage probability is calculated as 
\begin{equation}
    \mathbb{P}_{\text{cov}}^{\text{avg}} 
    = \frac{1}{K} \sum_{u=1}^{K} \mathbb{P} \left( {{\Lambda }_{u}} \right).
\end{equation}

\begin{theorem}
(Spectral efficiency): In the CoMS-NOMA network, the spectral efficiency of UT$_u$ is 
\begin{align}
    & \text{SE}_u  \nonumber \\
    & = \int_{0}^{\infty} \int_{0}^{{{R}_{\text{T}}}} \int_{{R_{\text{Ser,max}}}}^{{{R}_{\text{max} }}} {\int_{0}^{\eta }} \int_{{{R}_{\text{SU},\min }}}^{{{R}_{\text{SU},\max }}} \nonumber \\
    & \left[ \sum\limits_{i=1}^{{{N}_{\text{CS}}}}{\sum\limits_{l=1}^{\kappa }{{{c}_{i,l}}{{\left( -\frac{{{r}_{i}}^{-\alpha }}{\kappa } \right)}^{l}}}}{{f}_{{{R}_{\text{SU}}}}}\left( {{r}_{i}};{{r}_{\text{N}}},{{r}_{\text{T}}},{{\varphi }_{\text{z}}} \right) \right. \nonumber \\
    & \left. 
    \cdot \sum \limits_{k=0}^{l-1} \frac{{{\left[ \left( I^{\text{inter}}+{{{\bar{\sigma }}}^{2}} \right){ \left( \underset{u\le v\le K}{\mathop{\max }}\,\frac{2^t-1 }{ \tilde{\xi}_{v} } \right) } \right]}^{k}}}{k!} {{\left( -1 \right)}^{k}} \frac{{{d}^{k}}\mathcal{L}_{{I}^{\text{inter}}}\left( s \right)}{d{{s}^{k}}} \right] \nonumber \\
    & {\qquad \qquad \quad \quad} 
    \cdot d r_i {{f}_{{{\varphi }_{\text{z}}}}}\left( {{\varphi }_{\text{z}}} \right) d{{\varphi }_{\text{z}}} {{f}_{{{R}_{\text{N}}}}} \left( {{r}_{\text{N}}} \right) d{{r}_{\text{N}}} {{f}_{{{R}_{\text{T}}},u}}\left( {{r}_{\text{T}}} \right) d{{r}_{\text{T}}} dt.
\label{Formula_SE_u_1}
\end{align}
%
\end{theorem}
\begin{proof}
According to a mathematical theory 
\begin{equation}
    \mathbb{E}\{X\}=\int_{0}^{\infty}\mathbb{P}\{X>x\}dx,
\end{equation}
for $X>0$, the relationship between coverage probability and spectral efficiency can be established, and the spectral efficiency of UT$_u$ is given by
\begin{equation}
\begin{aligned}
    \text{SE}_u
    & = \int_{0}^{\infty} \mathbb{P} \left\{ \bigcap\limits_{v=u}^{K}{\left\{ \text{SINR}_{u}^{v}> 2^t-1 \right\}} \right\} dt.
\end{aligned}
\end{equation}
where $\gamma = 2^t-1$. 
Finally, applying ${{Q}_{u}} = \underset{u\le v\le K}{\mathop{\max }}\,\frac{\gamma }{ \tilde{\xi}_{v} } = \underset{u\le v\le K}{\mathop{\max }}\,\frac{2^t-1 }{ \tilde{\xi}_{v} }$ and an integral over $t$ in \eqref{Formula_P_Lambda_u_4}, 
the spectral efficiency of UT$_u$ is obtained in \eqref{Formula_SE_u_1}, completing the proof.
\end{proof}

To evaluate the overall system performance in terms of spectral efficiency, the sum spectral efficiency of the served UTs is calculated as
\begin{equation}
    \text{SE}^{\text{sum}} = \sum_{u=1}^{K} \text{SE}_u.
\end{equation}

\subsection{Benchmark Schemes Description}
By pairing satellite serving modes, either CoMS or SglS, with access schemes, i.e., NOMA or OMA, the effectiveness of the proposed CoMS-NOMA approach is evaluated against three benchmark schemes: (a) CoMS-OMA, (b) SglS-NOMA, and (c) SglS-OMA. 
Specifically, time-division multiple access (TDMA) is adopted as the representative scheme for the OMA-based benchmarks.
Correspondingly, the related expressions for SINR, coverage probability, or spectral efficiency of these benchmark schemes are briefly presented below to illustrate the analytical procedure.

\subsubsection{CoMS-OMA}
In OMA-based cooperative transmission, the UTs are served cooperatively by multiple satellites without inducing intra-satellite interference, as SIC is not required in this scenario.
The SINR at UT$_u$ of the message intended for UT$_v$ is represented as 
\begin{equation}
\begin{aligned}
    \text{SINR}_{u,\text{OMA}}^{v,\text{CoMS}}
    & = \frac{\sum\limits_{i\in {{\mathcal{S}}_{\text{CS}}}}{{{\left| {{h}_{u,i}} \right|}^{2}}{{\xi }_{v,i}}}} { \sum\limits_{j\in {{\mathcal{I}}_{\text{IS}}}}{I_{j}^{\text{inter}}}+{{{\bar{\sigma }}}^{2}}}.
\end{aligned}
\end{equation}
By substituting $Q_u$ in \eqref{Formula_P_Lambda_u_4} with $\gamma$, the coverage probability of UTs based on CoMS-OMA can be obtained.

Regarding spectral efficiency, the serving satellites cooperatively serve each UT$_u$ ($1 \leq u \leq K$) in turn in a time-division manner, where each UT is allocated a time slot $t_i = \frac{1}{K}$. 
In this case, by replacing ${{Q}_{u}} = \underset{u\le v\le K}{\mathop{\max }}\,\frac{2^t-1 }{ \tilde{\xi}_{v} }$ in \eqref{Formula_SE_u_1} with ${{Q}_{u}} = 2^t-1$, the spectral efficiency of CoMS-OMA UTs can be derived.

\subsubsection{SglS-NOMA}
By selecting the nearest satellite $i$ as the serving satellite and treating the remaining satellites as interfering ones, the SINR at UT$_u$ of the message intended for UT$_v$ is already shown in \eqref{Formula_SlgS_NOMA_SINR_1} or \eqref{Formula_SlgS_NOMA_SINR_2}.
Then, the coverage probability at UT$_u$ is written as 
\begin{equation}
\begin{aligned}
     \mathbb{P}_{\text{NOMA}}^{\text{SglS}} 
    & =\mathbb{P}\left( \bigcap\limits_{v=u}^{K}{\left\{ \text{SINR}_{u}^{v,i}>\gamma  \right\}} \right) \\
    & =\mathbb{P} \left( \bigcap\limits_{v=u}^{K} \left\{ \frac{ {{{\left| {{h}_{u,i}} \right|}^{2}}{{\xi }_{v,i}}} }{  {I_{i}^{\text{intra}}} + \sum\limits_{j\in {{\mathcal{I}}_{\text{IS}}}}{I_{j}^{\text{inter}}}+{{{\bar{\sigma }}}^{2}} } > \gamma  \right\} \right).
\end{aligned}
\end{equation}

\subsubsection{SglS-OMA}
Based on the expressions in SglS-NOMA, the SINR at UT$_u$ of the message intended for UT$_v$ in SglS-OMA is obtained by removing the intra-satellite interference term, i.e., ${I_{i}^{\text{intra}}}$. 
Then, the coverage probability is given by
\begin{equation}
\begin{aligned}
     \mathbb{P}_{\text{OMA}}^{\text{SglS}} 
    & =\mathbb{P} \left( \bigcap\limits_{v=u}^{K} \left\{ \frac{ {{{\left| {{h}_{u,i}} \right|}^{2}}{{\xi }_{v,i}}} }{  \sum\limits_{j\in {{\mathcal{I}}_{\text{IS}}}}{I_{j}^{\text{inter}}}+{{{\bar{\sigma }}}^{2}} } > \gamma  \right\} \right).
\end{aligned}
\end{equation}

The above derivations are straightforward based on the simulation and analysis of CoMS-NOMA and are therefore omitted for brevity to avoid repetitiveness.

\section{Simulations and Numerical Results}
In this section, we first evaluate the coverage performance of CoMS-NOMA networks and present Monte Carlo simulation results to validate the analytical results and investigate the impact of key system parameters.
Both coverage and spectral efficiency are compared with the increasing NOMA power allocation coefficient to demonstrate the superiority of CoMS-NOMA over representative benchmark schemes.

\begin{table*}[t]
\setlength{\abovecaptionskip}{0cm}
\captionsetup{font={scriptsize}}
\caption{Default Parameters Setup}
\label{parameter}
\begin{center}
\small
\begin{tabular}{m{5.8cm}<{} m{2.2cm}<{} m{5.8cm}<{} m{2.2cm}<{}}
\hline
\hline
\textbf{Description and Notation} & \textbf{Value} & \textbf{Description and Notation} & \textbf{Value}  \\
\hline
\hline
Radius of Earth, $R_\text{E}$ &  6,371 km   &
Altitude of satellites, $R_\text{S}$ &  500 km \\
Satellite transmit power, $P$ & 50 dBm &
Path-loss coefficient, $\alpha$ & 2.0 \\
Satellite main-lobe gain, $G_\text{ml}$ & 30 dBi &
Satellite side-lobe gain, $G_\text{sl}$ &  10 dBi \\
TN serving area radius, $R_{\text{T}}$ & 200 km  &
Receiver noise power, ${\sigma}^2$ & -110 dBm \\
Carrier frequency, $f_c$ & 1 GHz &
Error propagation factor, $\varpi$ & 0.1   \\
\hline
\hline
\end{tabular}
\vspace{-6mm} 
\end{center}
\end{table*}

\subsection{Simulation Setup}

The default parameters and their corresponding values are summarized in Table \ref{parameter}. These settings are mainly based on \cite{li2026coverage,li2019non,elhalawany2022outage,shang2024multi,li2026downlink}, and may be adjusted according to different scenarios of interest.
Moreover, following the characteristics of PPP, with an average number of $N = 3,000$ satellites distributed at the same orbital layer, the satellite density is calculated as $\lambda_{\text{S}} = \frac{N}{4\pi {R_{\text{S}}}^2}$. 
By varying the serving region angle $\eta$, the number of cooperative satellites can be controlled. The CoMS (with three satellites) and SglS cases are realized by setting $\eta=35^{\circ}$ and $\eta=21^{\circ}$, respectively.
Three UTs are considered, and the power allocation coefficients $\left[\xi_1,\xi_2,\xi_3\right]$ for UT$_1$, UT$_2$, UT$_3$ are set to $\left[1/6, 1/3, 1/2\right]$.

\captionsetup{font={scriptsize}}
\begin{figure}[tp]
\begin{center}
\setlength{\abovecaptionskip}{+0.2cm}
\setlength{\belowcaptionskip}{-0.0cm}
\centering
  \includegraphics[width=3.0in, height=2.3in]{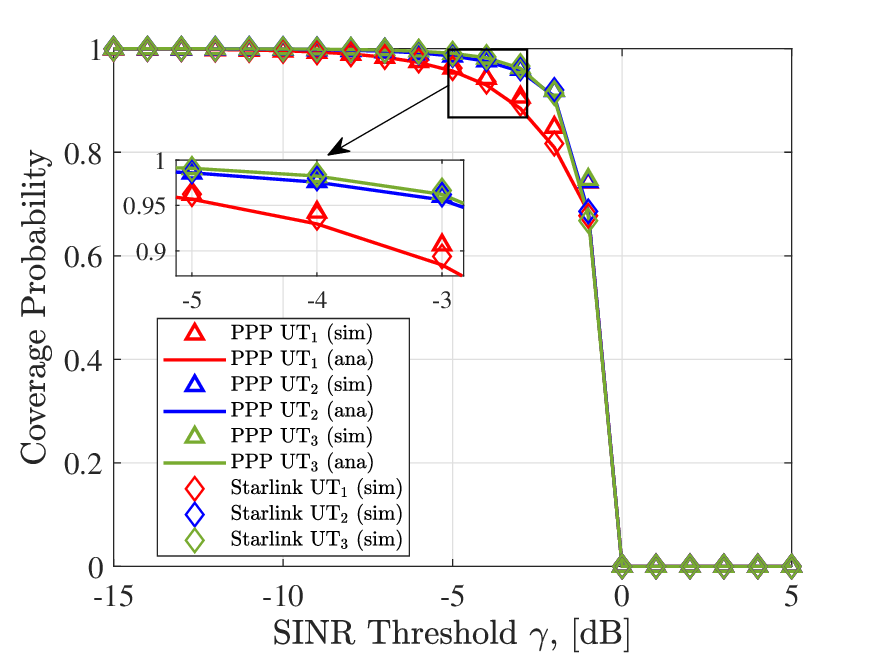}
\renewcommand\figurename{Fig.}
\caption{Coverage probability versus SINR threshold for different constellations.}
\label{Exp_fig6}
\end{center}
\vspace{-7mm}
\end{figure}

\captionsetup{font={scriptsize}}
\begin{figure}[tp]
\begin{center}
\setlength{\abovecaptionskip}{+0.2cm}
\setlength{\belowcaptionskip}{-0.0cm}
\centering
  \includegraphics[width=3.0in, height=2.3in]{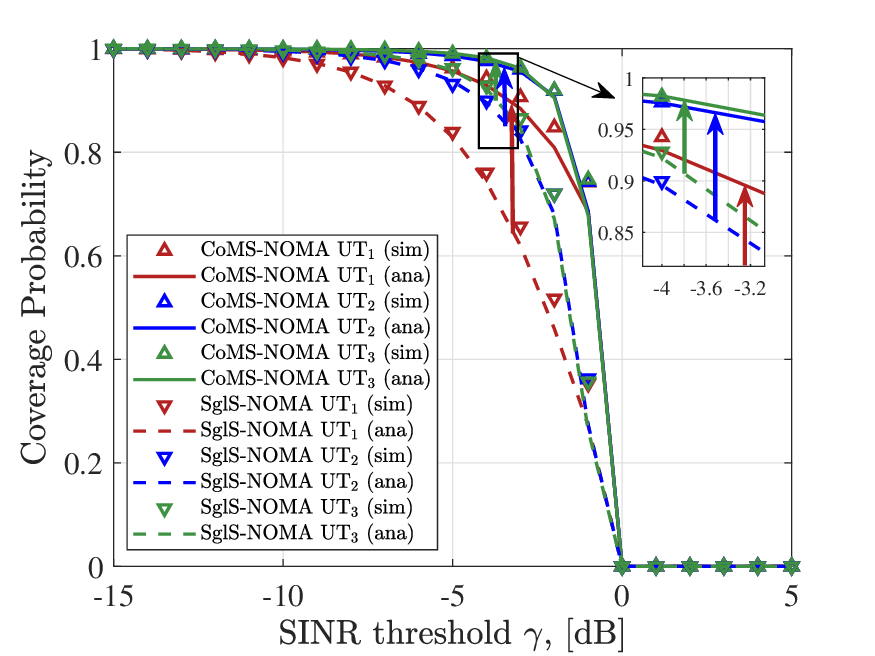}
\renewcommand\figurename{Fig.}
\caption{Coverage probability versus SINR threshold for SglS-NOMA, CoMS-NOMA}
\label{Exp_fig7}
\end{center}
\vspace{-7mm}
\end{figure}

\subsection{Coverage and Spectral Efficiency}

Fig. \ref{Exp_fig6} compares the coverage probability of three UTs versus the SINR threshold $\gamma$ under two different constellation models, i.e., the PPP-based constellation model and the actual Starlink constellation model. 
The analytical coverage expressions closely align with the Monte Carlo simulation results, verifying their accuracy.
In addition, the PPP-based constellation model can capture the statistical behavior of the Starlink constellation with an inclination of 53$^{\circ}$.
Hence, the satellite distribution is assumed to be independent of the PPP. 
Furthermore, the coverage probabilities of all three UTs are observed to decrease and converge to zero at the same SINR threshold. This is attributed to the NOMA necessary condition, which ensures successful operation only when all UTs can be simultaneously decoded. Once this condition is violated for any UT, NOMA transmission ceases.
Moreover, a higher coverage performance is achieved with larger power allocation coefficients for UTs. This indicates that coverage is predominantly influenced by power allocation, rather than marginally shortened satellite-UT distances.

The comparison between CoMS-NOMA and SglS-NOMA is illustrated in Fig. \ref{Exp_fig7}, where both analytical and simulation results are presented for three UTs. 
The close match between the analytical and simulation curves further validates the accuracy of the developed theoretical model. 
It can be observed that CoMS-NOMA consistently achieves a higher coverage probability across almost all SINR thresholds than SglS-NOMA. These higher achievements benefit from cooperative signal enhancement enabled by cooperative satellites. 
In particular, the minimum coverage probability of the UTs served by CoMS-NOMA surpasses the maximum coverage probability of those served by SglS-NOMA when $\gamma > -5$ dB. This indicates that CoMS-NOMA can maintain superior coverage performance even under stringent SINR requirements, and also demonstrates the enhanced reliability and robustness of CoMS-NOMA over SglS-NOMA systems.

\captionsetup{font={scriptsize}}
\begin{figure}[tp]
\begin{center}
\setlength{\abovecaptionskip}{+0.2cm}
\setlength{\belowcaptionskip}{-0.0cm}
\centering
  \includegraphics[width=3.0in, height=2.3in]{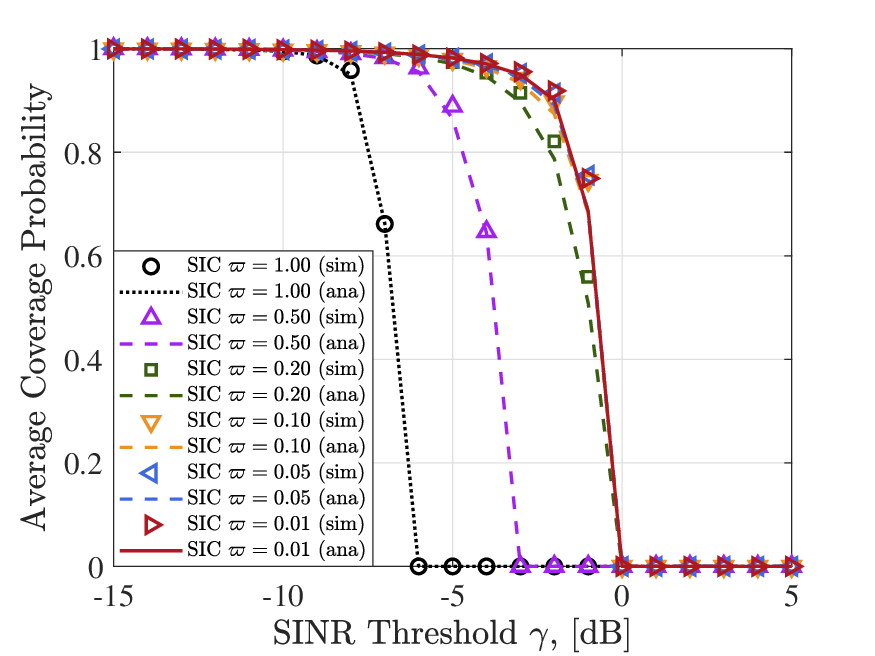}
\renewcommand\figurename{Fig.}
\caption{Coverage probability versus SINR threshold under different values of error propagation factor.}
\label{Exp_fig8}
\end{center}
\vspace{-7mm}
\end{figure}

\captionsetup{font={scriptsize}}
\begin{figure}[tp]
\begin{center}
\setlength{\abovecaptionskip}{+0.2cm}
\setlength{\belowcaptionskip}{-0.0cm}
\centering
  \includegraphics[width=3.0in, height=2.3in]{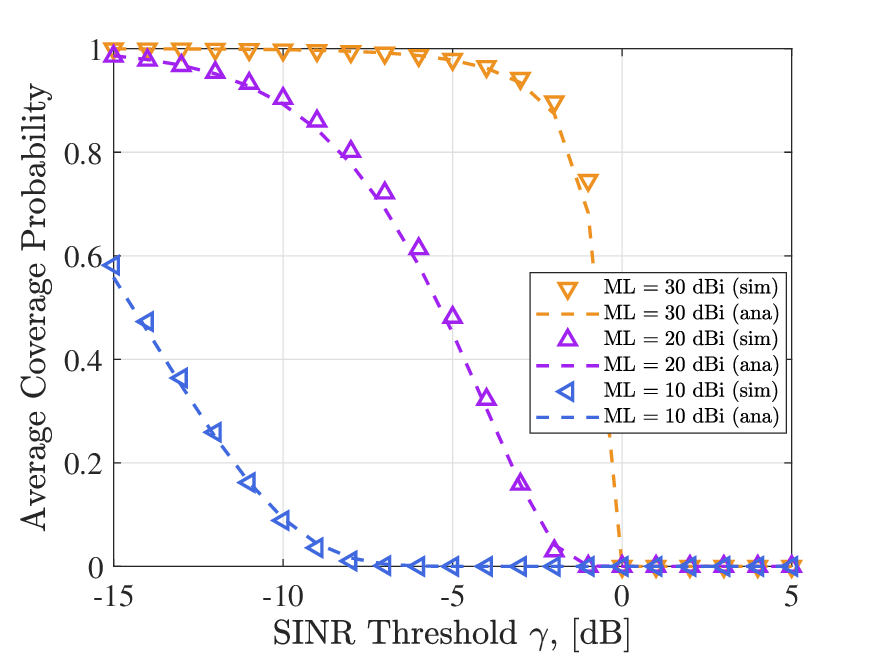}
\renewcommand\figurename{Fig.}
\caption{Coverage probability versus SINR threshold under different main-lobe gains.}
\label{Exp_fig9}
\end{center}
\vspace{-7mm}
\end{figure}

Fig. \ref{Exp_fig8} shows the impact of SIC conditions, characterized by the error propagation factor $\varpi$, on the average coverage probability of UTs. A series of $\varpi$ values, including $\varpi=1, 0.5, 0.2, 0.1, 0.05, 0.01$, are evaluated.
It is observed that under the worst-case decoding condition, i.e., when no SIC is performed by $\varpi = 1$, the average coverage probability is substantially lower than that achieved under other $\varpi$ values.
As $\varpi$ decreases, which corresponds to reduced residual interference under more advanced SIC receiver implementations, the average coverage probability correspondingly increases. 
In contrast, larger $\varpi$ values represent stronger residual error propagation caused by imperfect SIC decoding in practical receivers.
However, when $\varpi$ falls below $0.2$, the improvement becomes marginal, indicating diminishing returns from further enhancing the SIC accuracy.
This implies that while the effectiveness of the SIC has a significant impact on coverage, achieving a near-perfect SIC, i.e., $\varpi < 0.1$, yields limited additional benefit once the quality of the SIC is high.

In Fig. \ref{Exp_fig9}, under the fixed satellite side-lobe gain, the influence of the satellite main-lobe (ML) gains is illustrated for $30$ dBi, $20$ dBi, and $10$ dBi, respectively. 
As can be seen, a higher main-lobe gain significantly improves the average coverage probability across all SINR thresholds, with the curve for $30$ dBi main-lobe gain consistently outperforming the others. 
Moreover, the point at which the average coverage probability converges to zero shifts to the right as the main-lobe gain increases. This indicates that systems with higher main-lobe gains can maintain coverage up to higher SINR thresholds before losing connectivity.

\captionsetup{font={scriptsize}}
\begin{figure}[tp]
\begin{center}
\setlength{\abovecaptionskip}{+0.2cm}
\setlength{\belowcaptionskip}{-0.0cm}
\centering
  \includegraphics[width=3.0in, height=2.3in]{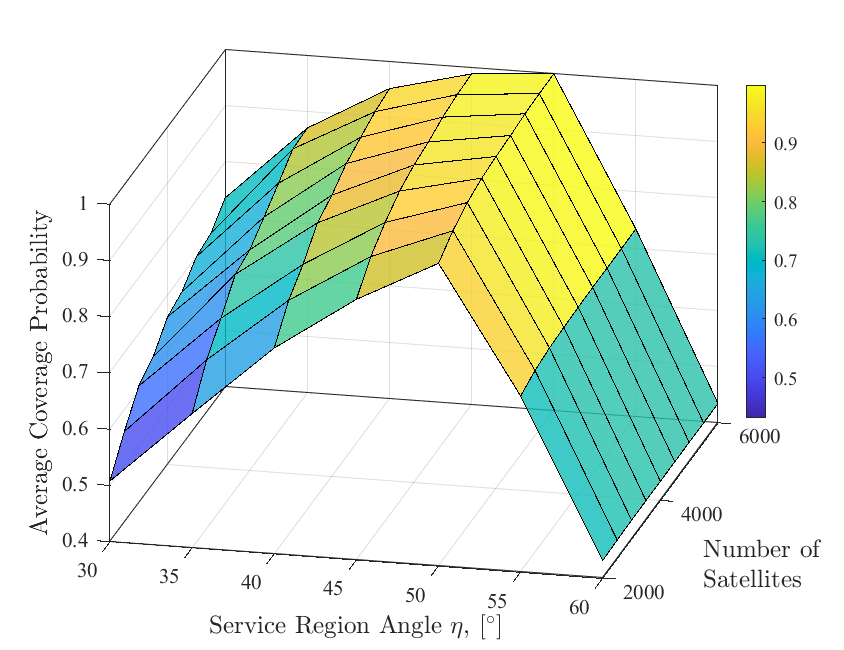}
\renewcommand\figurename{Fig.}
\caption{Coverage probability under different service region angles and total numbers of satellites, with $\gamma = 2$ dB and a user density of $\lambda_{\text{U}} = 1 \times 10^{-6}$/km$^2$.}
\label{Exp_fig10}
\end{center}
\vspace{-7mm}
\end{figure}

\captionsetup{font={scriptsize}}
\begin{figure}[tp]
\begin{center}
\setlength{\abovecaptionskip}{+0.2cm}
\setlength{\belowcaptionskip}{-0.0cm}
\centering
  \includegraphics[width=3.0in, height=2.3in]{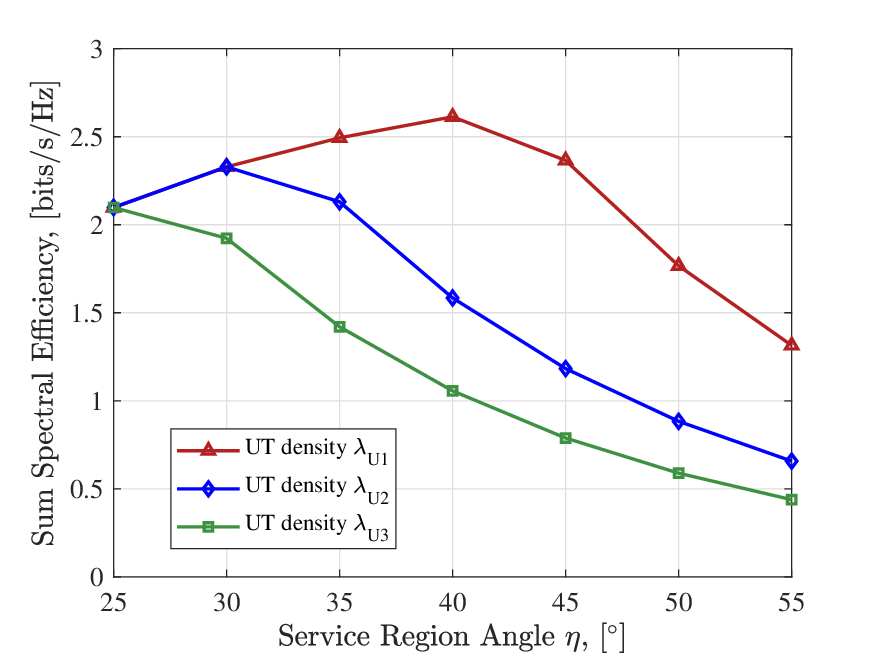}
\renewcommand\figurename{Fig.}
\caption{Coverage probability under different densities of UTs, where $\lambda_{\text{U1}} = 1 \times 10^{-6}$/km$^2$, $\lambda_{\text{U2}} = 2 \times 10^{-6}$/km$^2$, and $\lambda_{\text{U3}} = 3 \times 10^{-6}$/km$^2$.}
\label{Exp_fig11}
\end{center}
\vspace{-7mm}
\end{figure}

Fig. \ref{Exp_fig10} illustrates the joint influence of the service region angle and the total number of satellites on the average coverage probability of UTs in a three-dimensional (3D) plot. 
Notably, in line with practical implementations, a realistic user scheduling scenario is considered: as the UTs' service region angle increases, the corresponding serving range of each satellite expands. 
While this broader coverage improves spatial accessibility, it simultaneously decreases the likelihood that the targeted UTs are selected for service by a cooperative satellite, thereby introducing a fairness trade-off.
Specifically, increasing the total number of satellites raises the number of cooperative satellites within the service region angle of the targeted UTs, which enhances the aggregated received signal power and improves coverage performance.
Conversely, when the service region angle is relatively small, i.e., below $50^\circ$, the average coverage probability grows with the service region angle due to the enlarged service area. 
However, beyond approximately $50^\circ$, the coverage performance begins to decline, as the expanded serving range of each cooperative satellite encompasses more potential UTs within its coverage. 
This increased user density reduces the probability of each target UT being scheduled for service, which causes a drop in the average coverage probability.
This observation suggests that the service region angle should be carefully selected in practical systems to balance cooperative transmission gain and user scheduling efficiency under different user densities.

Beyond coverage probability, we now evaluate the system's spectral efficiency with user scheduling considerations. 
Fig. \ref{Exp_fig11} plots the sum spectral efficiency against the service region angle under three different user densities, i.e., $\lambda_{\text{U1}}$, $\lambda_{\text{U2}}$, and $\lambda_{\text{U3}}$. 
As observed, the sum spectral efficiency initially increases with the service region angle, reaching an optimum point, after which further expansion leads to a gradual decline.
This is because a larger service region angle increases the number of cooperative satellites within the serving region, which strengthens the aggregated desired signal power and consequently enhances the overall spectral efficiency in the initial stage.
However, beyond a certain angle, the satellite’s transmission resources are shared among more UTs other than NOMA UTs, reducing the scheduling priority of each individual and leading to a decrease in the sum spectral efficiency.
The optimal service region angle decreases with increasing user density and finally disappears for $\lambda_{\text{U3}}$, as the benefit of having more cooperative satellites is offset by the reduced serving probability due to more non-NOMA UTs served.
Therefore, practical CoMS-NOMA systems should adopt moderate power allocation configurations to balance user fairness, SIC reliability, and overall spectral efficiency.

\captionsetup{font={scriptsize}}
\begin{figure}[tp]
\begin{center}
\setlength{\abovecaptionskip}{+0.2cm}
\setlength{\belowcaptionskip}{-0.0cm}
\centering
  \includegraphics[width=3.0in, height=2.2in]{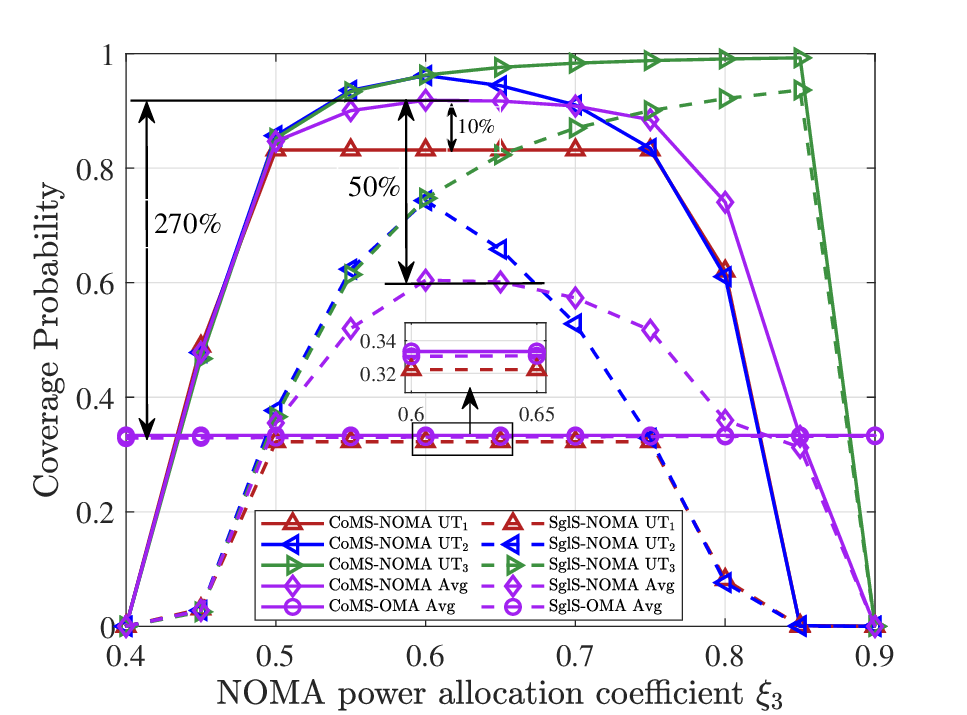}
\renewcommand\figurename{Fig.}
\caption{Coverage probability versus NOMA power allocation coefficient $\xi_3$, with SINR threshold $\gamma = -2.5$ dB.}
\label{Exp_fig12}
\end{center}
\vspace{-8mm}
\end{figure}

\captionsetup{font={scriptsize}}
\begin{figure}[tp]
\begin{center}
\setlength{\abovecaptionskip}{+0.2cm}
\setlength{\belowcaptionskip}{-0.0cm}
\centering
  \includegraphics[width=3.0in, height=2.2in]{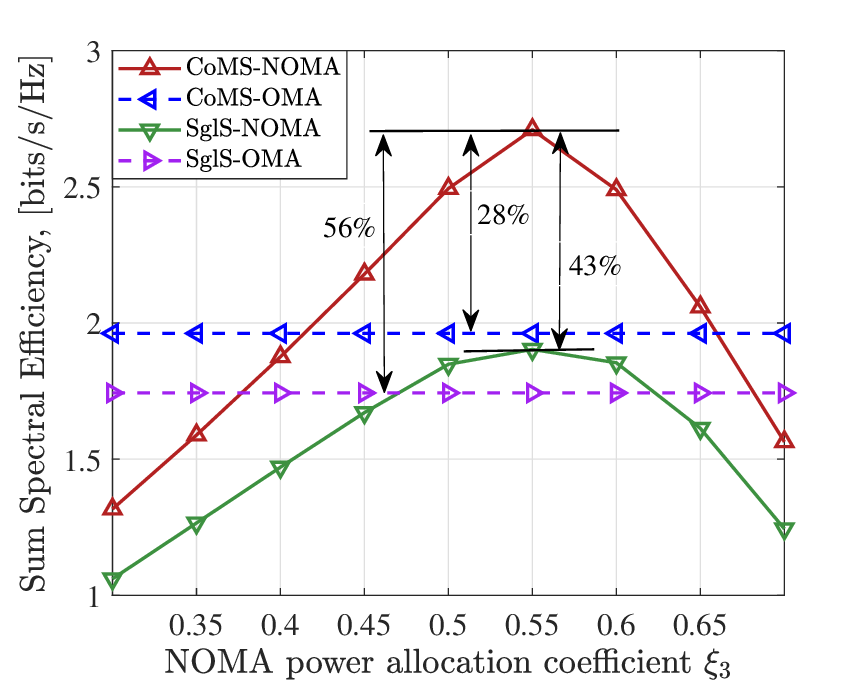}
\renewcommand\figurename{Fig.}
\caption{Sum spectral efficiency of UTs versus NOMA power allocation coefficient $\xi_3$, with $\xi_1 = 1/6$, $\xi_2 = 1-\xi_1-\xi_3$.}
\label{Exp_fig13}
\end{center}
\vspace{-8mm}
\end{figure}

\subsection{Comparison of NOMA and OMA}

Fig. \ref{Exp_fig12} compares the coverage probability of CoMS-NOMA with CoMS-OMA, SglS-NOMA, and SglS-OMA under a target SINR threshold of $\gamma=-2.5$ dB. 
Overall, CoMS-NOMA achieves superior coverage performance across a broad range of power allocation coefficients $\xi_3$, providing about a 50\% coverage gain over SglS-NOMA when $0.5 < \xi_3<0.8$, and exceeding CoMS-OMA and SglS-OMA by about 270\% when $0.43 < \xi_3<0.85$ in coverage. 
These results indicate that while CoMS-NOMA consistently outperforms CoMS-OMA, an appropriate selection of NOMA power allocation coefficients is crucial to guarantee better coverage gains over OMA counterparts; otherwise, the benefits of adopting NOMA become negligible.
Specifically, for CoMS-NOMA, the coverage probability of UT$_3$ continues to improve as its allocated power $\xi_3$ increases.
When $\xi_3 < 0.6$, this power adjustment yields a win–win situation for all three UTs. However, as $\xi_3$ further increases, both UT$_1$ and UT$_2$ experience a gradual decline in coverage, which eventually drops to zero. This is because excessive power imbalance impairs SIC operations and degrades overall NOMA transmission performance.

Fig. \ref{Exp_fig13} presents the sum spectral efficiency performance versus the NOMA power allocation coefficient $\xi_3$. The sum spectral efficiency of CoMS-NOMA first increases with $\xi_3$, reaching its maximum around $\xi_3 = 0.55$, and then decreases as an excessive power imbalance degrades the performance of the other UTs. A similar but less evident trend is observed for SglS-NOMA, with an overall lower sum spectral efficiency as a result of the absence of cooperative gain.
In contrast, both CoMS-OMA and SglS-OMA maintain nearly constant sum spectral efficiency values throughout $\xi_3$, as the sum spectral efficiency in OMA is insensitive to power allocation variations. 
Specifically, under the current parameter setting, CoMS-NOMA provides up to 28\%, 43\%, and 56\% improvements in sum spectral efficiency over CoMS-OMA, SglS-NOMA, and SglS-OMA, respectively.
In general, CoMS-NOMA achieves the highest spectral efficiency among all schemes within a wide range of $\xi_3$, highlighting the spectral efficiency advantage brought about by cooperative multi-satellite transmission and power-domain multiplexing.

\section{Conclusion}
In this paper, we investigated the downlink performance of CoMS-NOMA networks from a system-level perspective. A stochastic geometry-based analytical framework was developed, where satellites were modeled as a PPP on a spherical cap and served UTs were uniformly distributed within the serving area. The analytical expressions for the coverage probability and spectral efficiency were derived, and the analytical results were validated through extensive Monte Carlo simulations.
Numerical results showed that CoMS-NOMA can achieve approximately 50\% coverage improvement over SglS-NOMA and up to 56\% sum spectral efficiency enhancement over conventional benchmark schemes under appropriate power allocation settings.
Moreover, achieving near-perfect SIC is unnecessary, as a moderately accurate SIC condition is sufficient for reliable performance, whereas a higher main-lobe gain significantly enhances system average coverage. 
These findings will provide guidance for interference-resilient and spectrum-efficient CoMS-NOMA transmission design.

\appendices
\section{Proof of Proposition 1}
As shown in Fig. \ref{fig:Angle}, the relationship between zenith angle ${{\varphi }_{\text{z}}}$ and dome angle ${{\varphi }_{\text{d}}}$ can be written as \cite{al2021analytic}
\begin{equation}
\begin{aligned}
    {{\varphi }_{\text{d}}} 
    &= {{\cot }^{-1}}\left( \frac{\cot {{\varphi }_{\text{z}}}+\sqrt{{{\kappa }^{2}}\left( 1+{{\cot }^{2}}{{\varphi }_{\text{z}}}-{{\kappa }^{2}} \right)}}{1-{{\kappa }^{2}}} \right) 
    \triangleq \psi \left( {{\varphi }_{\text{z}}} \right),
\end{aligned}
\label{Formula:varphi_d_1}
\end{equation}
where $\kappa =\frac{{{R}_{\text{E}}}}{{{R}_{\text{E}}}+{{H}_{\text{S}}}} \in (0,1)$. Herein, satellite locations are modeled as a homogeneous PPP on the spherical surface of radius ${{R}_{\text{S}}}={{R}_{\text{E}}}+{{H}_{\text{S}}}$, and satellites within the service region angle $\eta$ are selected as cooperative satellites. 
By the conditional uniformity property of the PPP, if a satellite exists in the UT's service region, correspondingly $0 \le \varphi_{\text{z}} \le \eta$, the dome angle $\varphi_{\text{d}}$ is uniformly distributed with respect to the spherical surface measure over an interval $I_{\varphi_{\text{d}}} =\left[ \psi \left( 0 \right),\psi \left( \eta  \right) \right]$.
Denote the service region on the sphere as $\mathcal{A}$, as shown in Fig. \ref{fig:Angle}. 

On a spherical surface of radius ${{R}_{\text{S}}}$, a point on the sphere is expressed in Cartesian coordinates as $\left( x,y,z \right)=\left( {{R}_{\text{S}}}\sin {{\varphi }_{\text{d}}}\cos \xi ,{{R}_{\text{S}}}\sin {{\varphi }_{\text{d}}}\sin \xi ,{{R}_{\text{S}}}\cos {{\varphi }_{\text{d}}} \right)$, where $\left( \xi ,{{\varphi }_{\text{d}}} \right)\in \mathcal{A}$, $\xi \in \left[ 0,2\pi  \right)$ is azimuthal and ${{\varphi }_{\text{d}}}$ is polar. As two directions are orthogonal on the sphere, the area element is $dS={{R}_{\text{S}}}^{2}\sin {{\varphi }_{\text{d}}}d\xi d{{\varphi }_{\text{d}}}$. Satellites are uniformly distributed over region $\mathcal{A}$ with respect to surface area measure, and the probability that a satellite falls in $dS$ is $\mathbb{P}\left( \left( \xi ,{{\varphi }_{\text{d}}} \right)\in dS \right)=\frac{dS}{{{S}_{\mathcal{A}}}}$, where ${{S}_{\mathcal{A}}}=\int{\int_{\mathcal{A}}{dS}}$ is the surface area of $\mathcal{A}$. This probability also equals the integral of the joint PDF over the element $\mathbb{P}\left( \left( \xi ,{{\varphi }_{\text{d}}} \right)\in dS \right)={{f}_{\xi ,{{\varphi }_{\text{d}}}}}\left( {{\varphi }_{\xi ,{{\varphi }_{\text{d}}}}} \right)d\xi d{{\varphi }_{\text{d}}}$. Thus, after cancellation, we have ${{f}_{\xi ,{{\varphi }_{\text{d}}}}}\left( {{\varphi }_{\xi ,{{\varphi }_{\text{d}}}}} \right)=\frac{{{R}_{\text{S}}}^{2}}{{{S}_{\mathcal{A}}}}\sin {{\varphi }_{\text{d}}}$, showing that the joint PDF is proportional to $\sin {{\varphi }_{\text{d}}}$ with constant $\frac{{{R}_{\text{S}}}^{2}}{{{S}_{\mathcal{A}}}}$.

Additionally, the surface area of the region $\mathcal{A}$ is expressed as
${{S}_{\mathcal{A}}} 
=\int{\int_{\mathcal{A}}{dS}}
= \int_{0}^{2\pi }{\int_{\psi \left( 0 \right)}^{\psi \left( \eta  \right)}{{{R}_{\text{S}}}^{2}\sin tdt}}  
=2\pi {{R}_{\text{S}}}^{2}\left[ \cos \left( \psi \left( 0 \right) \right)-\cos \left( \psi \left( \eta  \right) \right) \right]  
=2\pi {{R}_{\text{S}}}^{2}\left[ 1-\cos \left( \psi \left( \eta  \right) \right) \right].$
Integrating over $\xi \in \left[ 0,2\pi  \right)$ azimuthally yields the marginal PDF of ${{\varphi }_{\text{d}}}$ shown as 
${{f}_{{{\varphi }_{\text{d}}}}}\left( {{\varphi }_{\text{d}}} \right) = \frac{\sin {{\varphi }_{\text{d}}}}{1-\cos \left( \psi \left( \eta  \right) \right)}$, where $\varphi_{\text{d}}\in I_{\varphi_{\text{d}}}$.
Recall the relationship ${{\varphi }_{\text{d}}}=\psi \left( {{\varphi }_{\text{z}}} \right)$ and the random variable transformation ${{f}_{{{\varphi }_{\text{z}}}}}\left( {{\varphi }_{\text{z}}} \right)={{f}_{{{\varphi }_{\text{d}}}}}\left( \psi \left( {{\varphi }_{\text{z}}} \right) \right) \left| \frac{d\psi \left( {{\varphi }_{\text{z}}} \right)}{d{{\varphi }_{\text{z}}}} \right|$ for $0\le {{\varphi }_{\text{z}}}\le \eta $. 
We now derive the derivative of $\psi \left( {{\varphi }_{\text{z}}} \right)$ over ${{\varphi }_{\text{z}}}$.

Let $q=\cot {{\varphi }_{\text{z}}}$, $\frac{dq}{d{{\varphi }_{\text{z}}}}=-{{\csc }^{2}}{{\varphi }_{\text{z}}}$, and $D\left( q \right)=\sqrt{{{\kappa }^{2}}\left( 1+{{q}^{2}}-{{\kappa }^{2}} \right)}$, so that $\frac{dD\left( q \right)}{dq}=\frac{{{\kappa }^{2}}q}{D\left( q \right)}$.
We also denote $\psi \left( {{\varphi }_{\text{z}}} \right)={{\cot }^{-1}}\left( J\left( q \right) \right)$, where $J\left( q \right)=\frac{q+D\left( q \right)}{1-{{\kappa }^{2}}}=\cot \left( \psi \left( {{\varphi }_{\text{z}}} \right) \right)$, which brings $1+{{J}^{2}}\left( q \right)={{\csc }^{2}}\left( \psi \left( {{\varphi }_{\text{z}}} \right) \right)$ based on \eqref{Formula:varphi_d_1}, and $\frac{d J\left( q \right)}{dq}=\frac{1}{1-{{\kappa }^{2}}}\left( 1+\frac{{{\kappa }^{2}}q}{D\left( q \right)} \right)$. Combining the above terms offers the derivative of $\psi \left( {{\varphi }_{\text{z}}} \right)$, which is written as
\begin{equation}
\begin{aligned}
    \frac{d\psi \left( {{\varphi }_{\text{z}}} \right)}{d{{\varphi }_{\text{z}}}}
    & =\frac{d{{\cot }^{-1}}\left( J\left( q \right) \right)}{d{{\varphi }_{\text{z}}}} \\
    & = -\frac{1}{1+{{J}^{2}}\left( q \right)}\frac{d J\left( q \right)}{dq}\frac{dq}{d{{\varphi }_{\text{z}}}}  \\
    & =\frac{{{\csc }^{2}}{{\varphi }_{\text{z}}}}{{{\csc }^{2}}\left( \psi \left( {{\varphi }_{\text{z}}} \right) \right)}\frac{1}{1-{{\kappa }^{2}}}\left( 1+\frac{{{\kappa }^{2}}q}{D\left( q \right)} \right).
\end{aligned}
\end{equation}
By inserting ${{f}_{{{\varphi }_{\text{d}}}}}\left( {{\varphi }_{\text{d}}} \right)$ and ${{\varphi }_{\text{d}}} \triangleq \psi \left( {{\varphi }_{\text{z}}} \right)$ of \eqref{Formula:varphi_d_1} into ${{f}_{{{\varphi }_{\text{z}}}}}\left( {{\varphi }_{\text{z}}} \right)={{f}_{{{\varphi }_{\text{d}}}}}\left( \psi \left( {{\varphi }_{\text{z}}} \right) \right) \left| \frac{d\psi \left( {{\varphi }_{\text{z}}} \right)}{d{{\varphi }_{\text{z}}}} \right|$, the PDF of ${{\varphi }_{\text{z}}}$ is finally represented as in \eqref{Formula:PDF_varphi_z}, which completes the proof.

\section{Proof of Lemma 2}
Due to the uniform distribution of the UTs in the area, angle $\zeta$ follows a uniform distribution, i.e., $\zeta \sim U\left( 0,2\pi  \right)$, so that ${{f}_{\zeta }}\left( \zeta  \right)=\frac{1}{2\pi }$ and ${{F}_{\zeta }}\left( \zeta  \right)=\frac{\zeta }{2\pi }$ for $\zeta \in \left[ 0,2\pi  \right]$. 
For variable transformation, denote ${{X}_{\zeta }}=\cos \zeta $, and we have $\frac{d{{X}_{\zeta }}}{d\zeta }=-\sin \zeta $, $\frac{d\zeta }{d{{X}_{\zeta }}}=-\frac{1}{\sin \zeta }$. Then, the PDF of ${{X}_{\zeta }}$ is 
\begin{equation}
\begin{aligned}
    {{f}_{{{X}_{\zeta }}}}\left( x \right)
    & ={{f}_{\zeta }}\left( \zeta  \right)\left| \frac{d\zeta }{dx} \right|  
     =\frac{1}{2\pi } \frac{1}{\sin \zeta }  
     =\frac{1}{2\pi \sqrt{1-{{x}^{2}}}}.
\end{aligned}
\end{equation}

Denote ${{X}_{\theta }}=\cos \theta $, then $\cos \zeta =\frac{{{X}_{\theta }}}{\sin {{\varphi }_{\text{z}}}}$, and $-\sin {{\varphi }_{\text{z}}}\le {{X}_{\theta }}\le \sin {{\varphi }_{\text{z}}}$. The PDF of ${X}_{\theta }$ is calculated as
\begin{equation}
\begin{aligned}
    {{F}_{{{X}_{\theta }}}}\left( x \right)
    & =\mathbb{P}\left\{ {{X}_{\theta }}\le x \right\}  
    =\mathbb{P}\left\{ \cos \zeta \le \frac{x}{\sin {{\varphi }_{\text{z}}}} \right\} \\
    & =\mathbb{P}\left\{ \zeta \ge \arccos \left( \frac{x}{\sin {{\varphi }_{\text{z}}}} \right) \right\}  \\
    & =1-\frac{1}{2\pi }\arccos \left( \frac{x}{\sin {{\varphi }_{\text{z}}}} \right),
\label{Formula:X_theta_PDF}
\end{aligned}
\end{equation}
and the PDF of ${X}_{\theta }$ is given by
\begin{equation}
\begin{aligned}
    {{f}_{{{X}_{\theta }}|{{\varphi }_{\text{z}}}}}\left( x|{{\varphi }_{\text{z}}} \right)
    & =\frac{d{{F}_{{{X}_{\theta }}}}\left( x \right)}{dx}  \\
    & =\frac{d}{dx}\left[ 1-\frac{1}{2\pi }\arccos \left( \frac{x}{\sin {{\varphi }_{\text{z}}}} \right) \right] \nonumber \\
    & =\frac{1}{2\pi \sqrt{{{\sin }^{2}}{{\varphi }_{\text{z}}}-{{x}^{2}}}}.
\end{aligned}
\end{equation}
Using the total probability formula, we have
\begin{equation}
\begin{aligned}
    {{f}_{{{X}_{\theta }}}}\left( x \right)
    &=\int_{{{\varphi }_{\text{z,lb}}}}^{{{\varphi }_{\text{z,ub}}}}{{{f}_{{{X}_{\theta }}|{{\varphi }_{\text{z}}}}}\left( x|{{\varphi }_{\text{z}}} \right){{f}_{{{\varphi }_{\text{z}}}}}\left( {{\varphi }_{\text{z}}} \right)d{{\varphi }_{\text{z}}}}  \\
    &=\int_{{{\varphi }_{\text{z,lb}}}}^{{{\varphi }_{\text{z,ub}}}}{\frac{{{f}_{{{\varphi }_{\text{z}}}}}\left( {{\varphi }_{\text{z}}} \right)}{2\pi \sqrt{{{\sin }^{2}}{{\varphi }_{\text{z}}}-{{x}^{2}}}}d{{\varphi }_{\text{z}}}}  \\
    & =\int_{0}^{\eta }{\frac{{{f}_{{{\varphi }_{\text{z}}}}}\left( {{\varphi }_{\text{z}}} \right)}{2\pi \sqrt{{{\sin }^{2}}{{\varphi }_{\text{z}}}-{{x}^{2}}}}d{{\varphi }_{\text{z}}}},
\end{aligned}
\end{equation}
where ${\varphi }_{\text{z,lb}} = 0$ and ${\varphi }_{\text{z,ub}} = \eta$ represent the lower bound and upper bound of $\varphi$, respectively.
Finally, consider ${{X}_{\theta }}=\cos \theta $, the PDF of angle $\theta$ is calculated as 
\begin{equation}
\begin{aligned}
    {{f}_{\theta }}\left( \theta  \right)
    & ={{f}_{{{X}_{\theta }}}}\left( \cos \theta  \right)\left| \frac{d}{d\theta }\cos \theta  \right|  \\
    & = \int_{0}^{\eta }{\frac{ \sin \theta \cdot {{f}_{{{\varphi }_{\text{z}}}}}\left( {{\varphi }_{\text{z}}} \right)}{\pi \sqrt{{{\sin }^{2}}{{\varphi }_{\text{z}}}-{{\cos }^{2}}\theta }}d{{\varphi }_{\text{z}}}},
\end{aligned}
\end{equation}
which completes the proof.

\section{Proof of Lemma 3}
First, due to the distribution of satellites in the service region and UTs in the serving area, the minimum distance between a satellite and a user is ${{R}_{\text{SU},\min }}={{H}_{\text{S}}}$, as shown in Fig. \ref{fig:Spherical_cap}.

Second, conditioned on the dome angle $\eta \in \left( 0,\frac{\pi }{2} \right)$, the height from the point $\left( 0,0,{{R}_{\text{E}}} \right)$ is \cite{li2026downlink}
${{H}_{\eta }}=\left( \sqrt{{{R}_{\text{E}}}^{2}+\frac{{{R}_{\text{S}}}^{2}-{{R}_{\text{E}}}^{2}}{{{\cos }^{2}}\eta }}-{{R}_{\text{E}}} \right){{\cos }^{2}}\eta$.
The complementary distance ${{H}_{\eta ,\text{c}}} = {{H}_{\text{S}}}-{{H}_{\eta }}$ is calculated as 
\begin{align}
    {{H}_{\eta ,\text{c}}}
    & = {{R}_{\text{S}}}-{{R}_{\text{E}}}{{\sin }^{2}}\eta -\sqrt{{{R}_{\text{E}}}^{2}+\frac{{{R}_{\text{S}}}^{2}-{{R}_{\text{E}}}^{2}}{{{\cos }^{2}}\eta }}{{\cos }^{2}}\eta.
\label{Formula:H_eta_c}
\end{align}
The link distance $r_\text{T}$ is measured from a random UT to the center of the serving area. Then, using the cosine theorem, the maximum distance is given by 
${{R}_{\text{SU},\max }}=\sqrt{{{r}_{\text{T}}}^{2}+{{\left( \frac{{{H}_{\eta }}}{\cos \eta } \right)}^{2}}-2{{r}_{\text{T}}}\frac{{{H}_{\eta }}}{\cos \eta } \cos \left( \eta +\frac{\pi }{2} \right)}$, where $\cos \left( \eta +\frac{\pi }{2} \right) = -\sin \eta$.
The distance range is ${{R}_{\text{SU},\min }}\le {{r}_{\text{SU}}}\le {{R}_{\text{SU},\max }}$ as shown in \eqref{Formula:range_r_SU}, which completes the proof.

\section{Proof of Proposition 2}
Using the CDF of ${{X}_{\theta }}=\cos \theta $, the CDF of ${{X}_{{{r}_{\text{SU}}}^{2}}}={{r}_{\text{SU}}}^{2}$ is
\begin{equation}
\begin{aligned}
    {{F}_{{{X}_{{{r}_{\text{SU}}}^{2}}}}}\left( {{x}_{{{r}_{\text{SU}}}^{2}}} \right)
    & = \mathbb{P}\left\{ {{X}_{{{r}_{\text{SU}}}^{2}}}\le {{x}_{{{r}_{\text{SU}}}^{2}}} \right\} \\
    & = \mathbb{P}\left\{ {{r}_{\text{N}}}^{2}+{{r}_{\text{T}}}^{2}-2{{r}_{\text{N}}}{{r}_{\text{T}}}\cos \theta \le {{x}_{{{r}_{\text{SU}}}^{2}}} \right\} \nonumber \\ 
    & \overset{\left( a \right)}{\mathop{=}} \mathbb{P}\left\{ \cos \theta \ge \Psi({{x}_{{{r}_{\text{SU}}}^{2}}}) \right\} \\
    & \overset{\left( b \right)}{\mathop{=}} 1-{{F}_{{{X}_{\theta }}}}\left( \Psi({{x}_{{{r}_{\text{SU}}}^{2}}}) \right) 
    =\frac{1}{2\pi }\arccos \left( \Psi({{x}_{{{r}_{\text{SU}}}^{2}}}) \right),
\end{aligned}
\end{equation}
where in (a) we denote $\Psi(t)=\frac{ {{r}_{\text{N}}}^{2}+{{r}_{\text{T}}}^{2}- t }{2{{r}_{\text{N}}}{{r}_{\text{T}}}\sin {{\varphi }_{\text{z}}}}$, and (b) follows from the PDF of ${X}_{\theta }$ in \eqref{Formula:X_theta_PDF}.
To obtain the CDF of ${X}_{{{r}_{\text{SU}}}}=\sqrt{{{X}_{{{r}_{\text{SU}}}^{2}}}}$, we follow the definition of the CDF of ${X}_{{{r}_{\text{SU}}}}$, i.e., 
$\mathbb{P}\left\{ \sqrt{{{X}_{{{r}_{\text{SU}}}}}}\le {{x}_{{{r}_{\text{SU}}}}} \right\} = \mathbb{P}\left\{ {{X}_{{{r}_{\text{SU}}}^{2}}}\le {{x}_{{{r}_{\text{SU}}}}}^{2} \right\} = {{F}_{{{X}_{{{r}_{\text{SU}}}^{2}}}}}\left( {{x}_{{{r}_{\text{SU}}}}}^{2} \right)$.
Conditioned on $r_\text{N}$, $r_\text{T}$, $\varphi_\text{z}$, the CDF of ${r}_{\text{SU}}$ is given by
\begin{align}
    {{F}_{{{R}_{\text{SU}}}}}\left( {{r}_{\text{SU}}};{{r}_{\text{N}}},{{r}_{\text{T}}},{{\varphi }_{\text{z}}} \right)
    & =\frac{1}{2\pi }\arccos \left( \frac{{{r}_{\text{N}}}^{2}+{{r}_{\text{T}}}^{2}-{{r}_{\text{SU}}}^{2}}{2{{r}_{\text{N}}}{{r}_{\text{T}}}\sin {{\varphi }_{\text{z}}}} \right).
\end{align}
Then, the PDF of ${r}_{\text{SU}}$ is given by
\begin{align}
    & {{f}_{{{R}_{\text{SU}}}}}\left( {{r}_{\text{SU}}};{{r}_{\text{N}}},{{r}_{\text{T}}},{{\varphi }_{\text{z}}} \right) \nonumber 
    =\frac{d{{F}_{{{R}_{\text{SU}}}}}\left( {{r}_{\text{SU}}};{{r}_{\text{N}}},{{r}_{\text{T}}},{{\varphi }_{\text{z}}} \right)}{d{{r}_{\text{SU}}}} \nonumber \\
    & =\frac{{{r}_{\text{SU}}}}{2\pi {{r}_{\text{N}}}{{r}_{\text{T}}}\sin {{\varphi }_{\text{z}}}\sqrt{1-{{\left( \frac{{{r}_{\text{N}}}^{2}+{{r}_{\text{T}}}^{2}-{{r}_{\text{SU}}}^{2}}{2{{r}_{\text{N}}}{{r}_{\text{T}}}\sin {{\varphi }_{\text{z}}}} \right)}^{2}}}},
\end{align}
which completes the proof.

\bibliographystyle{IEEEtran}
\bibliography{references.bib}

\end{document}